\documentclass[sigplan,screen,10pt]{acmart}
\usepackage{cleveref}
\title[Composing Verification of Foreign Functions with Cogent]{
	Overcoming Restraint: Composing Verification of Foreign Functions with Cogent}
\author{Louis Cheung}
\affiliation{%
	\institution{University of Melborune}
	\city{Melbourne}
	\state{Victoria}
	\country{Australia}
}
\email{lfcheung@student.unimelb.edu.au}
\author{Liam O'Connor}
\affiliation{%
	\institution{University of Edinburgh}
	\city{Edinburgh}
	\country{Scotland}
}
\email{l.oconnor@ed.ac.uk}
\author{Christine Rizkallah}
\affiliation{%
	\institution{University of Melbourne}
	\city{Melbourne}
	\state{Victoria}
	\country{Australia}
}
\email{christine.rizkallah@unimelb.edu.au}



\usepackage{tikz}
  \usetikzlibrary{calc}
  \usetikzlibrary{shadows}
  \usetikzlibrary{cd}
  \usetikzlibrary{shapes}
\usepackage{subcaption}
\usepackage{graphicx}
\graphicspath{{./imgs/}}

\usepackage{stmaryrd}
\usepackage{xspace}
\usepackage[linesnumbered,lined,commentsnumbered]{algorithm2e}

\usepackage{style/util}
\usepackage{style/cogentfigs}
\usepackage{listings}

\definecolor{commentcol}{rgb}{0.3,0.3,0.3}
\definecolor{keywordcol}{rgb}{0,0,0.4}
\definecolor{typecol}{rgb}{0.4,0.1,0}
\definecolor{funccol}{rgb}{0.1,0.4,0}
\definecolor{anticol}{rgb}{0.1,0.4,0.4}
\definecolor{light-gray}{gray}{0.90}
\lstdefinelanguage{systemsC}{
  language=C,
  morekeywords={size_t, $ty, $id, $exp, $spec},
}

\lstdefinelanguage{Cogent}{
  basicstyle=\ttfamily\footnotesize,
  sensitive=true,
  morecomment=[l]{--},
  morecomment=[s]{\{-}{-\}},
  morekeywords={type, let, in, and, if, then, else},
}

\lstdefinestyle{isa}{
  basicstyle=\ttfamily\footnotesize,  
  columns=fullflexible,
  keepspaces,
  mathescape=true,
  morecomment=[s]{(*}{*)},
  morekeywords={locale,fixes,assumes,shows,and,lemma,definition,in,
    type_synonym,where,procedures,theorem,if,then,case,of,let},
  literate=
    {"}{}0
    {ÔøΩ}{$^\prime$}1
    {'}{$^\prime$}1
    {\\<^sub>6}{$_6$}1      
    {\\<^sub>4}{$_4$}1      
    {\\<^sub>3}{$_3$}1      
    {\\<^sub>2}{$_2$}1      
    {\\<^sub>1}{$_1$}1      
    {\\<^sub>0}{$_0$}1      
    {\\<^sub>f}{$_f$}1
    {\\<^sub>G}{$_G$}1      
    {\\<^sub>T}{$_T$}1
    {\\<forall>}{$\forall$}1
    {\\<exists>}{$\exists$}1
    {\\<equiv>}{$\equiv$}1
    {\\<Longrightarrow>}{$\Longrightarrow$}2
    {\\<longrightarrow>}{$\longrightarrow$}2
    {\\<Rightarrow>}{$\Rightarrow$}1
    {\\<rightarrow>}{$\rightarrow$}1
    {\\<leftarrow>}{$\leftarrow$}1
    {\\<longleftrightarrow>}{$\longleftrightarrow$}2
    {\\<Down>}{$\Downarrow$}2
    {\\<guillemotleft>}{$\mathopen{\{\mkern-4mu|}$}2
    {\\<guillemotright>}{$\mathclose{|\mkern-4mu\}}$}2
    {\\<And>}{$\bigwedge$}1
    {\\<and>}{$\land$}1
    {\\<or>}{$\lor$}1
    {\\<not>}{$\lnot$}1
    {\\<in>}{$\in$}1
    {\\<notin>}{$\notin$}1
    {\\<noteq>}{$\neq$}1
    {\\<le>}{$\le$}1
    {\\<ge>}{$\ge$}1
    {\\<cup>}{$\cup$}1
    {\\<cap>}{$\cap$}1
    {\\<lambda>}{$\lambda$}1
    {\\<times>}{$\times$}1
    {\\<turnstile>}{$\vdash$}1
    {\\<turnstile-t>}{$\vdash_t$}2
    {\\<subseteq>}{$\subseteq$}1
    {\\<sqinter>}{$\sqcap$}1
    {\\<lbrace>}{$\mathopen{\{\mkern-4mu|}$}1
    {\\<rbrace>}{$\mathclose{|\mkern-4mu\}}$}1
    {\\<lbrakk>}{$\mathopen{\lbrack\mkern-3mu\lbrack}$}1
    {\\<rbrakk>}{$\mathclose{\rbrack\mkern-3mu\rbrack}$}1
    {\\<infinity>}{$\infty$}1
    {\\<sigma>}{$\sigma$}1
    {\\<mu>}{$\mu$}1
    {\\<gamma>}{$\gamma$}1
    {\\<tau>}{$\tau$}1
    {\\<Gamma>}{$\Gamma$}1
    {\\<le>}{$\leq$}1
    {\\<Sum>}{$\sum$}1
    {\\<dollar>}{\$}1
    }

\lstdefinelanguage{ffiHaskell}{
  language=haskell,
  morekeywords={foreign, ccall, unsafe},
}

\makeatletter
\newcommand{\removelatexerror}{\let\@latex@error\@gobble}
\makeatother

\newcommand{\ignore}[1]{}  % FIXME: remove
\newcommand{\sset}[1]{\{ #1 \}}

\newcommand{\ie}{i.e.,\ }

\newcommand{\dottvar}[1]{\ifx\relax#1\else
  \expandafter\ifx\string_#1\string_\allowbreak\else#1\fi
  \expandafter\dottvar\fi}
\newcommand{\ttvar}[1]{\texttt{\expandafter\dottvar\detokenize{#1}\relax}}

\newcommand{\typeparam}[1]{\textbf{\textcolor{blue}{#1}}}
\newcommand{\CDSL}{\Cogent}

\newcommand{\extfs}{\texttt{ext2}}
\newcommand{\bilby}{BilbyFs\xspace}
\newcommand{\walen}{\CFunName{length}\xspace}
\newcommand{\waget}{\CFunName{get}\xspace}
\newcommand{\waput}{\CFunName{put}\xspace}
\newcommand{\wafold}{\CFunName{fold}\xspace}
\newcommand{\wamap}{\CFunName{mapAccum}\xspace}

\newcommand{\Walen}[1]{\CFunName{length\textsubscript{#1}}\xspace}
\newcommand{\Waget}[1]{\CFunName{get\textsubscript{#1}}\xspace}
\newcommand{\Waput}[1]{\CFunName{put\textsubscript{#1}}\xspace}
\newcommand{\Wafold}[1]{\CFunName{fold\textsubscript{#1}}\xspace}
\newcommand{\Wamap}[1]{\CFunName{mapAccum\textsubscript{#1}}\xspace}
\newcommand{\RepeatAtm}[1]{\CFunName{repeat\textsubscript{#1}}\xspace}

\setcopyright{rightsretained}
\acmPrice{15.00}
\acmDOI{10.1145/3497775.3503686}
\acmYear{2022}
\copyrightyear{2022}
\acmSubmissionID{poplws22cppmain-p28-p}
\acmISBN{978-1-4503-9182-5/22/01}
\acmConference[CPP '22]{Proceedings of the 11th ACM SIGPLAN International Conference on Certified Programs and Proofs}{January 17--18, 2022}{Philadelphia, PA, USA}
\acmBooktitle{Proceedings of the 11th ACM SIGPLAN International Conference on Certified Programs and Proofs (CPP '22), January 17--18, 2022, Philadelphia, PA, USA}
\begin{document}

\begin{abstract}
\Cogent is a restricted functional language designed to
reduce the cost of developing verified systems code. Because 
of its sometimes-onerous restrictions, such as the lack of support for recursion 
and its strict uniqueness type system, \Cogent provides an escape hatch in the form of a
foreign function interface (FFI) to C code.
This poses a problem when verifying \cogent programs, as imported C components do not enjoy the same level of static guarantees
that \cogent does. Previous verification of file systems implemented in \cogent merely assumed that their C components were 
correct and that they preserved the invariants of \cogent's type system. 
In this paper, we instead prove such obligations. We demonstrate how they smoothly compose with existing \cogent theorems, and result in a correctness theorem of the overall 
\Cogent-C system. The \Cogent FFI constraints ensure that key invariants of \Cogent's type system are maintained even when calling C code. We verify reusable higher-order and polymorphic functions including a generic loop combinator and array iterators and demonstrate their application to several examples including binary search and the \bilby file system.
 We demonstrate the feasibility of verification of mixed \Cogent-C systems, and provide some insight into verification of 
software comprised of code in multiple languages with differing levels of static guarantees. 
\end{abstract}

\maketitle

\section{Introduction}
\label{sec:introduction}

\cogent~\cite{OConnor_CRALMNSK_16} is a restricted purely functional language with a \emph{certifying compiler}~\cite{Rizkallah_LNSCOMKK_16,OConnor_CRALMNSK_16} designed to ease creating verified operating systems components~\cite{Amani_HCRCOBNLSTKMKH_16}. It has a foreign function interface (FFI) that enables implementing parts of a system in C. \cogent's main restrictions are the purposeful lack of recursion or loops, which ensures totality,
and its \emph{uniqueness type system}, which enforces a \emph{uniqueness invariant} that, among other benefits, guarantees memory safety.

Even in the restricted target domains of \cogent, real programs contain some amount of iteration, primarily over data structures such as buffers.
This is achieved through \CDSL's principled FFI:
engineers provide data structures and their associated operations, including iterators,
in a special dialect of C, and import them into \CDSL, including in formal reasoning.
This special C code, called \emph{template C}, can refer to \cogent data types and functions, and is translated into standard C along with the \cogent program by the \cogent compiler.
As long as the C components respect \cogent's foreign function interface ---
\ie are correct, memory-safe and respect the uniqueness invariant ---
the \cogent framework guarantees that correctness properties proved on high-level specs also apply to the compiler output.
%\jashank{This paragraph is a hot mess.
%   Why did we suddenly talking about data structures?
%   Suggest explaining what the iteration is for
%   (``real programs contain some amount of iteration,
%   primarily over specific data structures.'')
%   as that will help to make this paragraph clearer.}

%Once these well-defined assumptions are discharged, the \cogent framework guarantees that the soundness of overall C-\CDSL\ system. 

%A core advantage of the \cogent framework is that
%it significantly reduces the effort required for verifying the \cogent components, 
%as one can equationally reason about these components.
Two real-world Linux file systems have been implemented in \cogent{} --- \extfs\ and \bilby\ --- and key operations of \bilby have been verified~\cite{Amani_HCRCOBNLSTKMKH_16}. This prior work demonstrates \Cogent's suitability as a systems programming language and as a verification framework that reduces  the cost of verification. 
%\vincent{What is this paragraph supposed to be saying? It just seems like a collection of "we did this".}
%
The implementations of these file systems import an external C library of data structures, which include fixed-length arrays and iterators for implementing loops, as well as \Cogent\ stubs for accessing a range of the Linux kernel's internal APIs. This library was carefully designed to ensure compatibility with \Cogent's FFI constraints,
but was previously left unverified --- only the \Cogent\ parts of these file system operations were proven correct, and statements of the underlying C correctness and FFI constraints defining \Cogent-C interoperability were left as assumptions.

To fully verify a system written in \Cogent\ and C, 
one needs to provide manually-written abstractions of the C parts,
and manually prove refinement through \cogent's FFI. 
The effort required for this manual verification remains substantial, but the
reusability of these libraries allows this cost to be amortised across different systems. 
%
%Realising \cogent's core vision hinges upon
%the ability to satisfy
%the conditions imposed by \cogent's FFI
%for the C components of the system.
%
%
%In producing the implementations of \extfs\ and \bilby, we also developed a library of abstract data types that includes fixed-length arrays for words and structures, simple iterators for implementing loops, and \Cogent\ stubs for accessing a range of the Linux kernel's internal APIs such as the buffer cache and its native red-black tree implementation. Those interfaces were carefully designed to ensure compatibility with \Cogent's FFI constraints. However, we had only verified the \Cogent\ parts of these file system operations: we axiomatized a small library of C functions, and the FFI constraints defining \Cogent-C interoperability.
% To fully verify a software system written in \Cogent\ and C,
% a proof engineer would need to provide manually-written abstractions of the C code, and manually prove the refinement theorems automatically generated for \Cogent code.
% As seen above, foreign functions tend to be reusable library functions. Thus, the cost in terms of verification effort of these functions can be amortised by reusing these manually-verified libraries in multiple systems. 
%% 
%-\christine{Christine: consider dropping or moving}
%+Realising \cogent's core vision hinges upon
%+the ability to satisfy
%+the conditions imposed by \cogent's FFI
%+for the C components of the system.
%
%\\[1em]
%\noindent
%\textbf{Contributions and Benefits}
%
%

In this paper, we eliminate several of these assumptions by verifying  
the array implementation and key iterators used in the \bilby file system (\Cref{sec:wordarray}), 
and discharging the conditions imposed by \Cogent's FFI (\Cref{sec:examples}). 
We also verify a generic-loop combinator (\Cref{sec:loops}) and its application to binary search.
This demonstrates that it is possible and relatively straightforward for the C 
components of a real-world \cogent-C system to satisfy \cogent's FFI conditions.
The compiler-generated refinement and preservation proofs
compose with manual C proofs at each intermediate level
up to \cogent's generated shallow embedding, and verify some example \Cogent-C programs. 

As arrays and loops are extremely common in \Cogent programming, our proofs are highly 
reusable for verification of any future \cogent-C system. In addition, our proofs connect C arrays 
to plain old Isabelle/HOL lists and loops to an Isabelle/HOL repeat function that allows early termination. 
These proofs are reusable even beyond the context of \cogent in verification of C code. 
Our code and proofs are online~\cite{cogent:wa-proofs}. %

%Our results show that \cogent's \emph{two-language} approach is not only suitable for developing real-world systems but also for verification. We show how to compose compiler generated theorems with external theorems. This demonstrates how  language interoperability can help guarantee correctness of an overall system. 
Similar to many high-level languages, Cogent's foreign function interface connects a high-level language with strong static guarantees to an unsafe imperative language.
This work provides our community with a case-study demonstrating how to equip such a foreign function interface with proof requirements such that those static guarantees are maintained for the overall system.
In particular, our work supports, and is well-described by, Ahmed's claim that:

\emph{
``Compositional compiler correctness is, in essence, a language interoperability problem: for viable solutions in the long term, high-level languages must be equipped with principled foreign-function interfaces that specify safe interoperability between high-level and low-level components, and between more precisely and less precisely typed code.''}~\cite{ahmed:LIPIcs:2016:5968} 

Our approach to language interoperability does not rely on how the refinement theorems of the  languages are obtained nor on whether they are manually or automatically proven. As such, we believe that this approach is likely reusable in the context of verified compilers.

\section{\cogent}
The \Cogent language~\cite{jfp,liam:phd,OConnor_CRALMNSK_16} was originally designed for the implementation of systems components such as file systems~\cite{Amani_HCRCOBNLSTKMKH_16}. It is a purely functional language, but it is compiled into efficient C code suitable for systems programming\footnote{While Cogent is ideally suited for applications that involve minimal sharing and where efficiency matters, it is not specific to the systems domain.}.

The \Cogent\ compiler produces three artefacts:
C code,
a shallow embedding of the \Cogent code in Isabelle/HOL~\cite{Nipkow_PW:Isabelle}, and
a formal refinement proof relating the two~\cite{OConnor_CRALMNSK_16,Rizkallah_LNSCOMKK_16}.
The refinement theorem and proof rely on several intermediate embeddings also generated by the \Cogent compiler,
some related through language level proofs, and others through translation validation phases (Section~\ref{ssec:cogent-verification}).
The compiler certificate guarantees that correctness theorems proven on top of the shallow embedding also hold for the generated C,
which eases verification,
and serves as the basis for further functional correctness proofs.
%This significantly reduces the verification effort of \cogent code. 
%Section~\ref{ssec:cogent-verification} gives an overview of the compiler certificate. 

%Uniqueness types enable statically tracking memory usage and allow for efficient destructive updates rather than the repeated copying commonly found in purely functional programming. Variables of a unique type must be used exactly once. This ensures a \emph{uniqueness invariant}: that each active mutable heap object has exactly one active pointer in scope at any point in the program.  

%This eliminate the need for a garbage collector.  
% Uniqueness types also provide a static means to track resources such as memory allocation, enabling a once and for all  proof that well-typed \cogent programs are memory safe (\autoref{fig:cogent}, arrow~(2)). Section~\ref{ssec:cogent-types} gives an overview of the type system. 
 %
 %
% This ensures a \emph{uniqueness invariant}: that each active mutable heap object has exactly one active pointer in scope at any point in the program.  
% Moreover, they allow modeling imperative computations as pure functions, as the type ensures that each mutable heap object has exactly one usable reference. 

A key part of the compiler certificate depends on \Cogent's \emph{uniqueness type system},
which enforces that each mutable heap object has exactly one active pointer in scope at any point in time.
This \emph{uniqueness invariant} allows modelling imperative computations as pure functions:
the allocations and repeated copying commonly found in functional programming
can be replaced with destructive updates, and the need for garbage collection is eliminated,
resulting in predictable and efficient code.

Well-typed \Cogent\ programs have two interpretations: a purely functional \emph{value semantics}, which has no notion of a heap and treats all objects as immutable values, and
an imperative \emph{update semantics}, describing the destructive mutation of heap objects. These two semantic interpretations correspond (Section~\ref{ssec:cogent-verification}), meaning that any correctness proofs about the value semantics also apply to the update semantics. As we shall see, this correspondence further guarantees that well-typed \cogent programs are memory safe.
%Section~\ref{ssec:cogent-semantics} gives an overview of cogent's semantics. 

%  Uniqueness types allow modeling imperative, stateful computations as pure mathematical functions, as the static type discipline ensures that each mutable heap object 
%has exactly one usable reference at any point in time. This means that a well-typed program can be given two interpretations:
%an imperative \emph{update} semantics that mutates heap objects and a ``pure'' \emph{value} semantics with no notion of a heap, that treats 
%all objects as immutable values. 
%
%As it is impossible to alias mutating objects, the equational reasoning by which we would reason about the pure interpretation applies 
%just as well to the imperative interpretation, as the lack of aliasing makes it impossible to observe the mutation of an object from
%any reference other than the one used to mutate it. Hence, well-typed \cogent programs are memory safe.

%\label{sec:overview}
% \begin{figure}
%%  \begin{center}
%    \includegraphics[scale=0.41, trim={0.2cm 5.5cm 22.5cm 0cm}]{cogent-adts-mem-detailed.pdf}
%%  \end{center}
%   \caption{An overview of the \Cogent framework (left) and the verification phases (right).}
%    \label{fig:cogent} 
%\end{figure}

\newcommand{\cLevel}{\textsf{c}}
\newcommand{\polyLevel}{\textsf{p}}
\newcommand{\monoLevel}{\textsf{m}}
\newcommand{\shallowLevel}{\textsf{s}}
\newcommand{\updateLevel}{\textsf{u}}
\definecolor{cRef1}{HTML}{ac92eb}
\definecolor{cRef2}{HTML}{4fc1e8}
\definecolor{updRef}{HTML}{a0d568}
\definecolor{monoRef}{HTML}{ffce54}
\definecolor{polyRef}{HTML}{ed5564}

%
%\christine{say something about FFI and arrows 4,5,6}
%

% As mentioned, the underlying C data structures of our two file systems were carefully designed to ensure
% compatibility with \Cogent's FFI constraints, but the verification of the \Cogent file systems previously assumed 
% this compatibility, as well as the correctness of the C code. Our verification discharges these assumptions, 
% giving us greater confidence in both the file systems and \cogent's FFI.

%Before we move to the formalization, we start by a small example of a \cogent program. \cogent 
\subsection{Language Design and Examples}

\cogent\ has unit, numeric, and boolean primitive types,
as well as functions, 
sum types (\ie variants) and product types (\ie tuples and records).
Users can declare additional \emph{abstract} types in \cogent,
and define them externally (\ie in C). Abstract and record types may be 
\emph{boxed}, that is, stored on the heap, in which case they are mutable
and subject to the uniqueness restrictions of \cogent's type system.
\Cogent\ does not support closures,
so partial application via currying is not common. Thus,
functions of multiple arguments
take a tuple or record of those arguments.

\begin{figure}
	\begin{tabular}{l}
		\CKwd{type} \CTypeName{Array} \CTypeVar{a}\\
\CFunName{length} : (\CTypeName{Array} \CTypeVar{a})\CBang\CFunctionArrow\CPrimType{U32}\\
%\CFunName{get} : ((\CTypeName{Array} \CTypeVar{a})\CBang, \CPrimType{U32}) \CFunctionArrow\CTypeVar{a}\\
\CFunName{get} : ((\CTypeName{Array} \CTypeVar{a})!, U32, a!) \CFunctionArrow \CTypeVar{a}!\\
\CFunName{put} : \CLParen\CTypeName{Array}  \CTypeVar{a}, \CPrimType{U32}, \CPrimType{a}\CRParen \CFunctionArrow \CTypeName{Array}  \CTypeVar{a}\\[0.5em]
\CFunName{map} : \CLParen\CTypeVar{a} \CFunctionArrow \CTypeVar{a}, \CTypeName{Array} \CTypeVar{a}\CRParen \CFunctionArrow \CTypeName{Array} \CTypeVar{a}\\
\CFunName{fold} : \CLParen \CLParen\CTypeVar{a}\CBang, \CTypeVar{b}\CRParen \CFunctionArrow \CTypeVar{b}, \CTypeVar{b}, (\CTypeName{Array} \CTypeVar{a})\CBang\CRParen \CFunctionArrow \CTypeVar{b}\\[0.5em]
\hline
$\ $\\[-0.6em]
\CFunName{add} : \CLParen\CPrimType{U32}, \CPrimType{U32}\CRParen \CFunctionArrow \CPrimType{U32}\\
\CFunName{add} \CLParen\CVar{x}, \CVar{y}\CRParen = \CVar{x} \CPrimOp{+} \CVar{y}\\
\CFunName{sum} : \CLParen\CTypeName{Array} \CPrimType{U32}\CRParen\CBang \CFunctionArrow \CPrimType{U32}\\
\CFunName{sum} \CVar{arr} = \CFunName{fold} \CLParen \CFunName{add}, \CLiteral{0}, \CVar{arr}\CRParen
	\end{tabular}
\caption{A \cogent \CFunName{sum} program that makes use of an abstract array type and operations.}
\label{ex:array}
\end{figure}
\Cref{ex:array}~includes an example of \Cogent signatures for an externally-defined array library interface, where array indices and length are unsigned 32-bit integers (\CPrimType{U32}).

Like ML, \Cogent supports \emph{parametric polymorphism} for top-level functions, 
and implements it via monomorphisation. For imported code,
the compiler generates specialised C implementations from a polymorphic template,
one for each concrete instantiation used in the \Cogent\ code.
Variables of polymorphic type are by default \emph{linear}, which means they must be used exactly once~\cite{Wadler_90}.
Thus a polymorphic type variable may be instantiated to any type, including types that 
contain pointers, while preserving the uniqueness invariant.
%Additional \emph{assumptions}, $A$, can be placed on the type variable, to restrict the possible instantiations to types that are less constrained. 
%can be safely shared ($\Shareable$), discarded ($\Discardable$), or returned from a let\CBang{} expression ($\Escapable$). 

%\subsection{Abstract types and functions}
As mentioned, types and functions provided in external C code are called \emph{abstract} in \cogent{}. The \cogent{} compiler has infrastructure for linking the C implementations and the compiled \cogent{} code.
Users write \emph{template C} code, that can include embedded \cogent types and expressions via quasi-quotation,
and the \cogent compiler translates the template C into ordinary C that it links with the C code generated from \Cogent.
To represent containers, abstract types may be given type parameters. %, such as in the \CTypeName{Array} type in~\autoref{fig:array}. 
These parameterised types, as well as polymorphic functions, are translated into a family of automatically generated C functions and types; one for each concrete type used in the \Cogent program.

Though the \CTypeName{Array} type interface may appear purely functional,
\cogent\ assumes that all abstract types are by default \emph{linear}, ensuring that 
the uniqueness invariant applies to 
variables of type \CTypeName{Array}.
Therefore, any implementation of the abstract \CFunName{put} function
is free to destructively update the provided array,
without contradicting the purely functional semantics of \cogent.

When functions only need to read from a data structure,
uniqueness types can complicate a program unnecessarily
by requiring a programmer to thread through all state,
even unchanged state.
The \CBang{}-operator helps to avoid this
by converting \emph{linear}, \emph{writable} types
to \emph{read-only} types that can be freely shared or discarded.
This is analogous to a \emph{borrow} in Rust.
The \CFunName{length} and \CFunName{get} functions, presented in \Cref{ex:array},  
can \emph{read} from the given array,
but may not \emph{write} to it.

%\cogent ensures that in \Cogent code,
%read-only references are never simultaneously live with writable references.
%This restriction is necessary
%to be able to reason equationally about \Cogent programs,
%and to preserve the refinement theorem connecting \cogent's two semantics.
%This condition is assumed of abstract functions by the \cogent semantics
%and must be proven of the C implementations for pointers that are visible to \Cogent.
%\vincent{
%  I rewrote this to make it a bit clearer; but it's also not true.
%  FFI functions can do whatever they want inside themselves,
%  the have to respect the frame condition,
%  which is slightly different to "read-only references are never
%  simultaneously live with wriable references".
% Liam: Fixed this to be true.
%}

As \cogent does not support recursion,
iteration is expressed using abstract \emph{higher-order functions},
providing basic traversal combinators
such as \CFunName{map} and \CFunName{fold} for abstract types, as 
can be seen in \Cref{ex:array}.
Note that \CFunName{map} is passed a function of type $\text{\CTypeVar{a}} \rightarrow \text{\CTypeVar{a}}$.
As such, \CFunName{map} is able to destructively overwrite the array with the result of the function applied to each element.

While \Cogent supports higher-order functions, 
it does not support nested lambda abstractions or closures,
as these can require allocation if they capture variables.
Thus, to invoke the \CFunName{map} or \CFunName{fold} functions,
a separate top-level function must be defined, such as \CFunName{add} in our example.

\begin{figure}
	\begin{tabular}{l}
\CFunName{length} : (\CTypeName{Array} \CTypeVar{a})! \CFunctionArrow \CPrimType{U32}\\
\CFunName{mapAccum} : 
 %\,\CForall \CTypeVar{a}, \CTypeVar{b} \CTypeVar{c}. 
\CLParen 
(\CTypeVar{a}, \CTypeVar{b}, \CTypeVar{c}\CBang) \CFunctionArrow (\CTypeVar{a}, \CTypeVar{b}),
\CTypeVar{b}, 
(\CTypeName{Array} \CTypeVar{a})\CBang, 
\CPrimType{U32}, 
\CPrimType{U32}, 
\CTypeVar{c}\CBang\CRParen\\
\quad\quad\quad\quad\CFunctionArrow (\CTypeName{Array} \CTypeVar{a}, \CTypeVar{b})\;
%fold: all(a :< DSE,acc,obsv). #{arr: (Array a)!, frm: U32, to: U32, 
\\%f: #{elem:a, acc:acc, obsv:obsv!} -> acc, acc: acc, obsv: obsv!} -> acc

\CFunName{fold} : 
% \,\CForall \CTypeVar{a}, \CTypeVar{b} \CTypeVar{c}. 
\CLParen
(\CTypeVar{a}, \CTypeVar{b}, \CTypeVar{c}\CBang) \CFunctionArrow \CTypeVar{b},  
\CTypeVar{b}, 
(\CTypeName{Array} \CTypeVar{a})\CBang,
\CPrimType{U32}, 
\CPrimType{U32}, 
\CTypeVar{c}\CBang\CRParen\\[0.5em]
\hline
$\ $\\[-0.6em]
\CFunName{add} : \CLParen\CPrimType{U32}, \CPrimType{U32}, \CLParen\CRParen\CRParen \CFunctionArrow \CPrimType{U32}\\
\CFunName{add} \CLParen\CVar{x}, \CVar{y}, \CVar{z}\CRParen = \CVar{x} \CPrimOp{+} \CVar{y}\\
\CFunName{sum} : (\CTypeName{Array} \CPrimType{U32})! \CFunctionArrow \CPrimType{U32}\\
%\CFunName{sum} \CVar{arr}  = \CIndent{ \CKwd{let} \CVar{len} = \CFunName{length} \CVar{arr}\;
%\CKwd{in} \CFunName{fold}[\CPrimType{U32}] \CLParen\CVar{arr}, \CLiteral{0}, \CVar{len}, \CFunName{add}, \CLiteral{0}, \CLParen\CRParen\CRParen}
\CFunName{sum} \CVar{arr} =
\CFunName{fold}
%\CLParen\CVar{arr}, \CLiteral{0}, \CFunName{length} \CVar{arr}, \CFunName{add}, \CLiteral{0}, \CLParen\CRParen\CRParen
\CLParen
\CFunName{add},  
0, 
\CVar{arr},
0, 
\CFunName{length} \CVar{arr}, 
\CLParen\CRParen\CRParen
	\end{tabular}
	\caption{The \CFunName{sum} function now written against the interace from \Cogent's C library.}
	\label{ex:iterators}
\end{figure}
The array interface from \Cogent's C library used in the implementation of the \Cogent file systems, part of which is given in~\Cref{ex:iterators}, is 
more complex than that of~\Cref{ex:array}: The higher-order functions are given two additional index parameters 
to operate over only a subsection of an array; and
instead of relying on closure captures, which are not available in \Cogent,
we provide alternative iterator functions which carry an
additional \emph{observer} read-only input (of type $\mathit{c!}$).
In addition, the function \CFunName{mapAccum} is a generalised version of \CFunName{map}, 
which allows threading an accumulating argument through the \CFunName{map} function, similar to
the same function in Haskell.
We present the verification of these iterators in~\Cref{sec:wordarray}.

\subsection{Dynamic Semantics}
\label{ssec:cogent-dynsem}
%\paragraph*{Dynamic Semantics}
%\begin{figure*}
% \centering \(
%  \begin{array}{llclr}
%    \text{v. sem. values} & v & \defd & l & (\textit{literals})\\
%		& & \alt & \FunVal{x}{e} & (\textit{function values - expressions e})\\
%		& & \alt & \AbsFunVal{f_{\mathit{id}}}{\overline\tau} & (\textit{abstract
%		functions - identifier } f_{\mathit{id}})\\
%		& & \alt & \overline{\RecordVal{\FieldEq{f}{v}}} & (\textit{records -
%		field name } \FieldN{f})\\
%    & & \alt & C\ v & (\textit{variants - constructor C})\\
%	  & & \alt & () \alt a_v & \\
%	  \text{v. sem. abstract values} & a_v & \defd & \textcolor{red}{\mathtt{VWA}\ \tau\
%	  \overline{v}} & (\textit{word arrays - constructor } \mathtt{VWA})\\
%	  & & \alt & \ldots & \\
%	  \hline
%    \text{u. sem. values} & u & \defd & l \alt \FunVal{x}{e} \alt
%	  \AbsFunVal{f_{\mathit{id}}}{\overline\tau} & \\
%	  & & \alt & \overline{\RecordVal{\FieldEq{u}{v}}} \alt C\ u \alt () \alt a_u & \\
%	  & & \alt & p & (\textit{pointers})\\
%	  \text{u. sem. abstract values} & a_u & \defd & \textcolor{red}{\mathtt{UWA}\ \tau\
%	  \mathtt{u32}\ p}  & (\textit{word arrays - constructor }
%	  \mathtt{UWA},\\
%	  & & & &  \textit{- 32-bit word } \mathtt{u32})\\
%	  & & \alt & \ldots & \\
%    \end{array}
%    \)\\   
%   % $\overline{\text{Overlines}}$ indicate lists.
%	\caption{Grammar for values in the \textit{value} (top) and
%	\textit{update} (bottom) semantics.}
%    \label{fig:values}
%\end{figure*}
\cogent's big-step \emph{value} semantics is defined through the judgement
\mbox{$\ValSemA{\xi_v}{V}{e}{v}$}. 
%Given a function $\xi_v: f_{\mathit{id}} \rightarrow (v \rightarrow v)$, that for an abstract
%function with identifier $f_{\mathit{id}}$ defines its value semantics,
This judgement states that the expression $e$  under environment $V$ evaluates to the value
$v$. The environment $V$ maps variables to their values.
The imperative \emph{update} semantics, which additionally may manipulate a mutable store $\mu$, is defined through the judgement  
\mbox{$\UpdSemAb{\xi_u}{U}{e}{\mu}{u}{\mu'}$}. This states that, starting with an initial 
store $\mu$ the expression $e$ will evaluate under the environment $U$ to a final store $\mu'$ and a result value $u$. Unlike 
the values in the value semantics, values in the update semantics may be \emph{pointers} to locations in the store.

Both of these semantics are further parameterised by additional \emph{functions} and types of \emph{values} that are 
provided externally to \Cogent, to model the semantics of abstract functions and types. 
More formally, the value semantics is parameterised by a function $\xi_\mathsf{v}: f_{\mathit{id}} \rightarrow (v \rightarrow v)$
and the update semantics by a function $\xi_\mathsf{u}:
f_{\mathit{id}} \rightarrow (\mu \times u \rightarrow \mu \times u)$. Both of these are essentially an environment providing a 
 pure HOL function on \Cogent values (and stores, for the update semantics) for each abstract function. The definitions of 
 values in the value semantics ($v$) and update semantics ($u$) are also extended with parameters $a_\mathsf{v}$ and $a_\mathsf{u}$ respectively 
which represent values of abstract types.

Along with C code for all abstract functions and types, the user must also manually provide 
Isabelle/HOL abstractions of this C code to instantiate these environments. 

To verify \Cogent systems, three main proof obligations must be discharged: \emph{type preservation}, which ensures the uniqueness invariant is maintained;
 the \emph{frame requirements}, which ensures that memory safety is maintained; and \emph{refinement}, which ensures that functional correctness theorems are preserved down to the 
 C level via the provided abstractions. \Cogent proves all three of these requirements automatically for \Cogent code: both type preservation and the frame requirements are simple corollaries of 
 the key semantic correspondence theorem (\Cref{thm:updvalrefinement}) that makes up part of the \Cogent refinement chain. For linked C code, however, the user must discharge these obligations manually.
 We discuss our verification of these requirements for C code in \Cref{sec:wordarray}.

\subsection{Typing and Type Preservation}
\label{ssec:type-pres}
\Cogent's \emph{static semantics} are defined through a standard typing judgement $\Typing{A}{\Gamma}{e}{\tau}$, which states that 
$e$ has type $\tau$ under context $\Gamma$,
with an additional context $A$ that tracks assumptions about the linearity of type variables in $\tau$. To accommodate 
abstract types, we allow the type system to be extended with types $\AbsTy{A}{\overline{\tau}}$ and $(\AbsTy{A}{\overline{\tau}})!$, referring 
to linear abstract types and read-only abstract types respectively, where \CPrimType{A} is a type constructor parameterised by zero or more type 
parameters $\overline{\tau}$.

Dynamic values in the value semantics are typed by the simple judgement $\VTRV{v}{\tau}$, whereas update semantics values must be typed with the store $\mu$ to type the parts of the value that are stored there.
Update semantics values are typed by the judgement $\VTRU{u}{\mu}{\tau}{r}{w}$,
which additionally includes the \emph{heap footprint},
consisting of the sets of read-only ($r$) and writable ($w$) pointers the value can contain. 
%Disjointness requirements are added to the definition of this judgement to ensure that the uniqueness invariant is maintained.
We use the same notation for value typing on environments.

This heap footprint annotation is crucial to ensuring that \Cogent maintains its uniqueness invariant,
as it places constraints on the footprints of subcomponents of a value to rule out aliasing of live pointers.
Thus, our theorem of type preservation
across evaluation in the update semantics also shows that this invariant is preserved.
More 
details on these constraints are discussed in earlier work~\citep{jfp,OConnor_CRALMNSK_16}. 
When the heap footprints are not relevant and merely existentially quantified, we shall omit them:\ $$\VTRUE{u}{\mu}{\tau}\ \equiv\ \exists r\ w.\ \VTRU{u}{\mu}{\tau}{r}{w}$$

\begin{theorem}[Type Preservation]\label{thm:presu}
	For both update and value semantics:\\[1em]
	$\begin{array}{lccccccl}
		\bullet & \Typing{A}{\Gamma}{e}{\tau}\!\!&\!\! \land\!\! &\!\!\VTRUE{U}{\mu}{\Gamma}\!\!
		&\!\! \land\!\!&\!\! \UpdSemAb{\xi_u}{U}{e}{\mu}{u}{\mu'}\!\! &\!\! \longrightarrow
		\!\!&\!\!\VTRUE{u}{\mu'}{\tau}\\[0.5em]
		\bullet & \Typing{A}{\Gamma}{e}{\tau}\!\!&\!\! \land\!\! &\!\!\VTRV{V}{\Gamma}\!\!
		&\!\! \land\!\!&\!\! \ValSemA{\xi_v}{V}{e}{v} &\!\! \longrightarrow
		\!\!&\!\!\VTRV{v}{\tau}\\		
	\end{array}$
\end{theorem}
\noindent This states that the value typing relation for either semantics is preserved across the evaluation relation for 
well-typed expressions. Because the value typing relations of the update and value semantics are later combined into one 
refinement relation, which is shown to be preserved across evaluation in~\Cref{thm:updvalrefinement}, type preservation
 is obtained by simply erasing one of the semantics from~\Cref{thm:updvalrefinement}.

Because the set of types is extensible, the value-typing relation for both semantics must also be 
extensible. To ensure that the user's extensions to the value-typing relation do not violate the uniqueness 
invariant, \Cogent places a number of proof obligations on abstract types that must be discharged by the user. 
These requirements are outlined in \Cref{ssec:type-reqs}.

\subsection{Frame Requirements}
\label{ssec:framerecs}
In addition to type preservation, which ensures that each \Cogent{} value is well-formed and 
does not contain internal aliasing, we must also show that the mutable store $\mu$ is 
in good order throughout evaluation --- memory should not be leaked, and programs should not write 
to memory to which they have no access. These \emph{memory safety} requirements are summed 
up by \Cogent's \emph{frame} relation, which describes how a program may affect the store.
\newcommand{\Frame}[4]{#1\ |\ #2\ \textbf{frame}\ #3\ |\ #4}
Given an input set of writable pointers $w_i$,
an input store $\mu_i$,
an output set of pointers $w_o$ and
an output store $\mu_o$,
the relation, $\Frame{w_i}{\mu_i}{w_o}{\mu_o}$,
ensures three properties for any pointer $p$:
\begin{align*}
  \textrm{inertia:\quad}
  & p \notin w_i \cup w_o  \longrightarrow
    \mu_i(p) = \mu_o(p) \\
  \textrm{leak freedom:\quad}
  & p \in w_i  \longrightarrow
    p \notin w_o  \longrightarrow
    \mu_o(p) = \bot \\
  \textrm{fresh allocation:\quad}
  & p \notin w_i \longrightarrow
    p \in w_o \longrightarrow
    \mu_i(p) = \bot
\end{align*}
\emph{Inertia} ensures that
pointers not in the frame remain unchanged;
\emph{leak freedom} ensures that
pointers removed from the frame no longer point to anything; and
\emph{fresh allocation} ensures that
pointers added to the frame were not already used.
The frame relation implies that
any property of a given value
is unaffected by updates to unrelated parts of the heap.
The frame relation holds for all \Cogent computations,
ensuring memory safety along with type safety:
\begin{theorem}[Preservation and Frame Relation]\label{thm:framepres}$\ $\\[0.5em]
	$\begin{array}{ll}
		\Typing{A}{\Gamma}{e}{\tau}\ \land\ \VTRU{U}{\mu}{\Gamma}{r}{w}\ 
		\land\ \UpdSemAb{\xi_u}{U}{e}{\mu}{u}{\mu'}\ \longrightarrow\\\
		\quad\exists r'\ w'.\ r' \subseteq r\ \land\
		\framerelargs{w}{\mu}{w'}{\mu'}\ \land\ \VTRU{u}{\mu}{\tau}{r'}{w'}
	\end{array}$
\end{theorem}
\noindent This states that a well-typed program will evaluate in the update 
semantics to a well-typed value, \emph{and} that the frame relation holds
between the writable pointers of the input environment and the output value.
Note that this theorem implies  update semantics type preservation
(\Cref{thm:presu}), because this theorem too is a simplification of \Cref{thm:updvalrefinement}.

\subsection{Refinement}
\label{ssec:cogent-verification}
\begin{figure}
\begin{tikzpicture}[y=1.1cm]
\draw[color=white,fill=gray!10!white] (-2.2,0.8) rectangle (6.2,3.3);	

\node[anchor=base] (c_heap) at (-0.38,0) {$\sigma$};
\node[anchor=base] at (-0.66,0) {\large$\vert$};
\node[anchor=base] (c_arg) at (-1.04,0) {$a_\cLevel{}$};
\node[anchor=base] (c_func) at (0.45,0) {$f_\cLevel{}$};
\node[anchor=base] (c_ret) at (1.5,0) {$r_\cLevel{}$};
\node[anchor=base] at (0,0) {\Large$\vdash$};
\node[anchor=base] at (1.8,0) {\large$\vert$};
\node[anchor=base] (c_heap2) at (2.2,0) {$\sigma'$};
\node[anchor=base,yshift=0.09cm,single arrow,single arrow head extend=0.05cm, draw,text height=0.04cm, shape border rotate=270, inner sep=0.1cm, label = {[label distance=-0.29cm]0:\scriptsize \sffamily c}] at (1,0) {};
\node[anchor=base] (upd_heap) at (-0.38,1) {$\mu$};
\node[anchor=base] at (-0.66,1) {\large$\vert$};
\node[anchor=base] (upd_env) at (-1.46,1) {$(x \mapsto a_\updateLevel{})$};
\node[anchor=base] (upd_func) at (0.45,1) {$f_\monoLevel{}$};
\node[anchor=base] (upd_ret) at (1.5,1) {$r_\updateLevel{}$};
\node[anchor=base] at (0,1) {\Large$\vdash$};
\node[anchor=base] at (1.8,1) {\large$\vert$};
\node[anchor=base] (upd_heap2) at (2.2,1) {$\mu'$};
\node[anchor=base,yshift=0.09cm,single arrow,single arrow head extend=0.05cm, draw,text height=0.04cm, shape border rotate=270, inner sep=0.1cm, label = {[label distance=-0.3cm]0:\scriptsize \sffamily u}] at (1,1) {};
\node[anchor=base] (mono_env) at (-0.95,2) {$(x \mapsto a_\monoLevel{})$};
\node[anchor=base] (mono_func) at (0.45,2) {$f_\monoLevel{}$};
\node[anchor=base] (mono_ret) at (1.5,2) {$r_\monoLevel{}$};
\node[anchor=base] at (0,2) {\Large$\vdash$};
\node[anchor=base,yshift=0.09cm,single arrow,single arrow head extend=0.05cm, draw,text height=0.04cm, shape border rotate=270, inner sep=0.1cm, label = {[label distance=-0.295cm]0:\scriptsize \sffamily v}] at (1,2) {};
\node[anchor=base] (poly_env) at (-0.929,3) {$(x \mapsto a_\polyLevel{})$};
\node[anchor=base] (poly_func) at (0.45,3) {$f_\polyLevel{}$};
\node[anchor=base] (poly_ret) at (1.5,3) {$r_\polyLevel{}$};
\node[anchor=base] at (0,3) {\Large$\vdash$};
\node[anchor=base,yshift=0.09cm,single arrow,single arrow head extend=0.05cm, draw,text height=0.04cm, shape border rotate=270, inner sep=0.1cm, label = {[label distance=-0.295cm]0:\scriptsize \sffamily v}] at (1,3) {};
\node[anchor=base] (shallow_arg) at (-0.6,4) {$a_\shallowLevel{}$};
\node[anchor=base] (shallow_func) at (0.45,4) {$f_\shallowLevel{}$};
\node[anchor=base] (shallow_ret) at (1.5,4) {$r_\shallowLevel{}$};
\node[anchor=base] at (0,4) {\Large$\vdash$};
\node[anchor=base,yshift=0.09cm,single arrow,single arrow head extend=0.05cm, draw,text height=0.04cm, shape border rotate=270, inner sep=0.1cm, label = {[label distance=-0.285cm]0:\scriptsize \sffamily s}] at (1,4) {};

\draw[color=cRef1!90!black,very thick] (c_heap.north) -- ([yshift=0.05cm]upd_heap.south);
\draw[color=cRef1!90!black,very thick] ([yshift=-0.1cm]c_heap2.north) -- ([yshift=0.05cm]upd_heap2.south);
\draw[color=cRef2!90!black,very thick] (c_ret.north) -- (upd_ret.south);
\draw[color=cRef2!90!black,very thick] (c_arg.north) -- ([yshift=0.65cm]c_arg.north);
\draw[color=updRef!90!black, very thick] (upd_ret.north) -- (mono_ret.south);
\draw[color=updRef!90!black, very thick] ([xshift=0.4cm,yshift=-0.1cm]upd_env.north) to[out=90,in=270] ([xshift=0.3cm,yshift=0.1cm]mono_env.south);
\draw[color=updRef!90!black, very thick] (upd_heap.north) to[out=90,in=270] ([xshift=0.3cm,yshift=0.1cm]mono_env.south);
\draw[color=updRef!90!black, very thick] ([yshift=-0.1cm]upd_heap2.north) to[out=90,in=270] (mono_ret.south);
\draw[color=monoRef!90!black, very thick] ([yshift=-0.06cm]mono_ret.north) -- (poly_ret.south);
\draw[color=monoRef!90!black, very thick] ([xshift=0.33cm,yshift=-0.12cm]mono_env.north) -- ([yshift=0.5cm,xshift=0.33cm]mono_env.north);
\draw[color=polyRef!90!black, very thick] ([xshift=0.3cm,yshift=-0.12cm]poly_env.north) to[out=90,in=270] ([xshift=0.3cm,yshift=0.55cm]poly_env.north);
\draw[color=polyRef!90!black, very thick] ([yshift=-0.05cm]poly_ret.north) -- (shallow_ret.south);

\node[anchor=east,yshift=0.05cm] at (6.2,0) {\small \sffamily AutoCorres C};
\node[anchor=east,yshift=0.05cm] at (6.2,1) {\small \sffamily Update Semantics};
\node[anchor=east,yshift=0.05cm] at (6.2,2) {\small \sffamily Value Semantics, Monomorphic};
\node[anchor=east,yshift=0.05cm] at (6.2,3) {\small \sffamily Value Semantics, Polymorphic};
\node[anchor=east,yshift=0.05cm] at (6.2,4) {\small \sffamily Shallow HOL Embedding};
\node[anchor=base] at (0.6, 3.5) {\footnotesize (\Cref{thm:scorres})};
\node[anchor=base] at (0.6, 2.5) {\footnotesize (\Cref{thm:monomorphise})};
\node[anchor=base] at (0.6, 1.5) {\footnotesize (\Cref{thm:updvalrefinement})};
\node[anchor=base] at (0.6, 0.5) {\footnotesize (\Cref{thm:corres})};
\end{tikzpicture}
\caption{\Cogent's semantic levels and refinement theorems.}
\label{fig:refinement}
\end{figure}
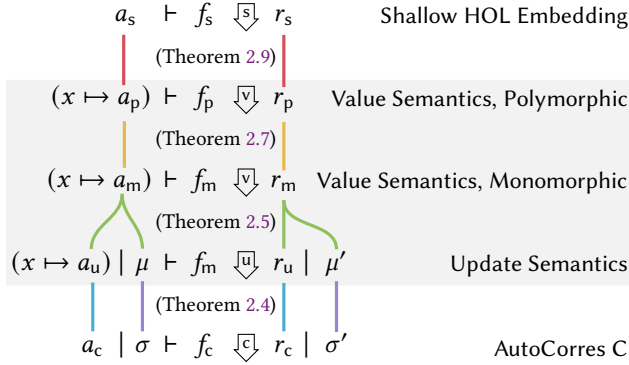
%\cogent has a \emph{certifying compiler} that apart from compiling a \cogent program to C, generates an Isabelle/HOL shallow embedding 
%of the program in terms of simple functions, and a proof that the generated C code is a \emph{refinement} of that embedding.
%This entails that any functional correctness theorem proven about the simple shallow embedding also applies to the generated C code.

%In prior work~\cite{Amani_HCRCOBNLSTKMKH_16}, we developed two Linux file systems in \cogent, and proved key correctness theorems about one of them. 
%The equational semantics drastically reduced the effort required to verify these systems, and their performance was comparable to other
%Linux file systems that were hand-written in C.
The overall proof that the C code refines the purely functional shallow embedding in Isabelle/HOL
is broken into a number of sub-proofs and translation validation phases.
\autoref{fig:refinement} gives an overview of \Cogent{}'s refinement theorems for a function $f(x)$ applied 
to an argument $a$.
The compiler generates four embeddings:
a top-level shallow embedding in terms of pure functions;
a polymorphic deep embedding of the \cogent program,
  which is interpreted using the \emph{value} semantics;
a monomorphic deep embedding of the \cogent program,
  which can be interpreted using either the \emph{value} or \emph{update} semantics;
and an Isabelle/HOL representation of the C code generated by the compiler,
  imported into Isabelle/HOL by the C-parser~\cite{Winwood_KSACN_09} 
  used in the seL4 project~\cite{Klein_EHACDEEKNSTW_09}.
The C-parser generates a deep embedding of C in Isabelle/HOL,
and, using AutoCorres~\cite{Greenaway_AK_12,Greenaway_LAK_14},
is then abstracted to a corresponding state-monadic embedding of the C code in HOL.

Each of these semantic layers is connected by a refinement proof by \emph{forward
simulation}: Given a \emph{refinement relation} that relates corresponding values 
between two layers, we prove that if the more \emph{concrete} (lower in the hierarchy)
layer evaluates, then the more \emph{abstract} (higher in the hierarchy) layer, given corresponding inputs, will
also evaluate to corresponding outputs. The composition of all these refinements
means that any property preserved by refinement,
such as functional correctness, proved about all executions of the most \emph{abstract} embedding --- the Shallow HOL embedding --- will also apply to the most \emph{concrete} embedding, i.e.\ the C code.

\subsubsection{Update to C Refinement}

In the first stage, \Cogent proves refinement between the C implementation and the deep
embedding in the update semantics. AutoCorres imports C as a nondeterministic 
state-monadic program 
shallowly embedded in Isabelle/HOL. To make our definitions more symmetrical with those of
\Cogent, we define a C evaluation relation as follows:
\begin{definition}[C Evaluation Relation]
	\[\begin{array}{ll}
		\CSem{a_\cLevel}{\sigma}{f_\cLevel}{r_\cLevel}{\sigma'}\ \equiv\\  
		\quad\ (r_\cLevel, \sigma') \in \results (f_\cLevel\ a_\cLevel\ \sigma) \land \lnot \failed (f_\cLevel\ a_\cLevel\ \sigma)
	\end{array}\]
\end{definition}
\noindent This states that given an input $a_c$ and C heap $\sigma$, the C function $f_c$ evaluates to $r_c$ and an output 
heap $\sigma'$ and that no undefined behaviour occurred (indicated by $\neg\mathit{failed}$).

\newcommand{\refrel}[4]{\textcolor{#1!80!black}{\mathcal{#2}}_{#3}^{#4}}
\newcommand{\valrelupdC}{\ensuremath{\refrel{cRef2}{V}{\cLevel}{\updateLevel}}\xspace}
\newcommand{\heaprelupdC}{\ensuremath{\refrel{cRef1}{H}{\cLevel}{\updateLevel}}\xspace}
\newcommand{\monorel}{\ensuremath{\refrel{monoRef}{R}{\monoLevel}{\polyLevel}}\xspace}
\newcommand{\shallowrel}{\ensuremath{\refrel{polyRef}{R}{\polyLevel}{\shallowLevel}}\xspace}

The \Cogent compiler additionally generates a \emph{value relation} \valrelupdC
and a \emph{heap relation} \heaprelupdC, which together form the refinement relation
for this refinement lemma. Using an automated technique described elsewhere~\citep{Rizkallah_LNSCOMKK_16}, 
the \Cogent proof then automatically discharges this proof obligation on a per-program basis 
via translation validation, while leaving open proof obligations for the user to discharge for abstract functions implemented in C.

\begin{theorem}[Update $\Rightarrow$ C refinement]\label{thm:corres}
For any function $f(x)$ with monomorphic \Cogent embedding $f_\monoLevel$ and C embedding $f_\cLevel$,
given an argument represented in the update semantics of \Cogent as $a_\updateLevel$ and in C as $a_\cLevel$, we have:
	\[\begin{array}{ll}
	\valrelupdC(a_\cLevel,a_\updateLevel) \land \heaprelupdC(\sigma,\mu)
		\land \CSem{a_\cLevel}{\sigma}{f_\cLevel}{r_\cLevel}{\sigma'} \longrightarrow\\
		\quad \exists r_u\ \mu'.\ \UpdSem{(x \mapsto a_\updateLevel)}{f_\monoLevel}{\mu}{r_\updateLevel}{\mu'} \\\qquad\qquad \land\ \valrelupdC(r_\cLevel,r_\updateLevel) \land \heaprelupdC(\sigma',\mu') 
	\end{array}\]
\end{theorem}
\noindent This states that 
if the C embedding evaluates to a result, then the corresponding \Cogent update semantics will, given 
corresponding input values and heaps, evaluate to corresponding output values and heaps.

\subsubsection{Semantic Correspondence}
For the second stage, we must bridge the gap between the \textit{update}
semantics and \textit{value} semantics. This is accomplished by \Cogent's proof 
for all well-typed programs that the two semantics correspond.
As previously mentioned, this theorem combines both of the value typing relations
from the two semantics $\VTRV{v}{\tau}$ and $\VTRU{u}{\mu}{\tau}{r}{w}$ into one 
 one combined relation $\VTRN{u}{\mu}{v}{\tau}{r}{w}$. In addition to typing both 
 values, $u$ and $v$, this relation also requires that they represent the same conceptual value.
 Prior work~\citep{jfp,OConnor_CRALMNSK_16} proves that the update semantics refines 
 the value semantics by proving that this relation is preserved across the evaluation of both, and furthermore
 that evaluation in the update semantics implies the evaluation in the value semantics. As mentioned in \Cref{ssec:framerecs},
 this also simultaneously proves the frame requirements hold for \Cogent code.

\begin{theorem}[Value $\Rightarrow$ Update refinement]
\label{thm:updvalrefinement}
For any $e$ where $\Typing{A}{\Gamma}{e}{\tau}$, if $\ \VTRN{U}{\mu}{V}{\Gamma}{r}{w}$ and $\UpdSem{U}{e}{\mu}{u}{\mu'}$, then there exists a value $v$ and pointer sets $r' \subseteq r$ and  $w'$ such that $\ValSem{V}{e}{v}$, and $\ \VTRN{u}{\mu'}{v}{\tau}{r'}{w'}$ and $\Frame{w}{\mu}{w'}{\mu'}$.
\end{theorem}
\noindent This proof is parameterised by the assumption that the same holds for abstract functions.
We discuss how to discharge this assumption in~\Cref{sec:wordarray}.
\subsubsection{Monomorphisation}
Recall that the \Cogent compiler eliminates polymorphism by \emph{monomorphising}, that is, 
replacing polymorphic functions with specialised copies, one for each type used in the program.
Using template C, the \Cogent compiler can do the same for the user-supplied C code. 
To prove that this elimination of polymorphism preserves correctness we must show that the monomorphic 
program refines the polymorphic program. This is accomplished by replicating the compiler's monomorphisation 
operations in Isabelle/HOL as functions on deep embeddings: $\mathcal{M}_e$ to monomorphise expressions and $\mathcal{M}_v$ to monomorphise values. 
Then, the refinement relation is simply defined:
\begin{definition}[Monomorphisation Relation] 
	$$\monorel(v_\monoLevel,v_\polyLevel) \equiv (v_\monoLevel = \mathcal{M}_v(N,v_\polyLevel))$$
where $N$ is a \emph{name mapping}, supplied by the compiler, indicating which set of monomorphic functions correspond to which polymorphic functions.
\end{definition}
\noindent Proving refinement, then, proceeds on much the same lines as the other layers:
\begin{theorem}[Polymorphic $\Rightarrow$ Monomorphic refinement]\label{thm:monomorphise}
	For any function $f(x)$ with a polymorphic embedding $f_\polyLevel$ and argument $a_\polyLevel$, let $f_\monoLevel = \mathcal{M}_e(N, f_\polyLevel)$
	and $a_\monoLevel = \mathcal{M}_v(a_\polyLevel)$.
	Then we have:
		\[\begin{array}{ll}
		 \monorel(a_\monoLevel,a_\polyLevel)
			\land \ValSem{(x \mapsto a_\monoLevel)}{f_\monoLevel}{r_\monoLevel} \longrightarrow\\
			\qquad\quad\quad\exists r_\polyLevel.\  \ValSem{(x \mapsto a_\polyLevel)}{f_\polyLevel}{r_\polyLevel} \land \monorel(r_\monoLevel,r_\polyLevel)
		\end{array}\]
	\end{theorem}

\subsubsection{Shallow Embedding}
Having reached the top of the refinement tower, we must cross the bridge back to shallow embeddings, \ie our pure
HOL functions that serve as an executable specification. Because this embedding is just pure functions, our 
``evaluation'' relation is just function application:
\begin{definition}[Shallow Evaluation Relation]
	$$
		\SSem{a_\shallowLevel}{f_\shallowLevel}{r_\shallowLevel}\ \equiv
		(f_\shallowLevel\ a_\shallowLevel = r_\shallowLevel)
	$$
\end{definition}
\noindent To show refinement, the compiler must once again connect deep and shallow embeddings, just as with 
the \Cogent to C refinement (\Cref{thm:corres}). Therefore, as with C, the compiler automatically produces a proof 
of this theorem on a per-program basis via translation validation. The compiler generates a refinement relation \shallowrel
for all types used in the program, and then proves:
\begin{theorem}[Shallow $\Rightarrow$ Polymorphic refinement]\label{thm:scorres}
	For any function $f(x)$ with a shallow embedding $f_\shallowLevel$ and polymorphic deep embedding $f_\polyLevel$,
	and arguments $a_\shallowLevel$ and $a_\polyLevel$ respectively, we have:
		\[\begin{array}{ll}
		\shallowrel(a_\polyLevel,a_\shallowLevel)
			\land \ValSem{(x \mapsto a_\polyLevel)}{f_\polyLevel}{r_\polyLevel} \longrightarrow\\
			\qquad\qquad\qquad\;\exists r_\shallowLevel.\ \SSem{a_\shallowLevel}{f_\shallowLevel}{r_\shallowLevel} \land \shallowrel(r_\polyLevel,r_\shallowLevel)
		\end{array}\]		
\end{theorem}

\subsection{Overall Refinement}
To combine all of these refinement phases, we first define a refinement relation 
for all the layers of refinement:
\begin{definition}[Combined Relation] 
	\[\begin{array}{l}
		\refrel{black}{R}{\cLevel}{\shallowLevel}(v_\cLevel,\sigma,v_\updateLevel,\mu,v_\monoLevel,v_\polyLevel,v_\shallowLevel,\tau,r,w) \\[0.2em]
		\qquad \equiv \shallowrel(v_\polyLevel,v_\shallowLevel)  \land \monorel (v_\monoLevel,v_\polyLevel) \\
		\qquad \land\ \VTRN{v_\updateLevel}{\mu}{v_\monoLevel}{\tau}{r}{w} \\[0.2em]
		\qquad \land\ \valrelupdC(v_\cLevel,v_\updateLevel) \land  \heaprelupdC(\sigma,\mu) 
	\end{array}\]
\end{definition}
\begin{theorem}[Combined Refinement]
	\label{thm:everything}
	Given a function $f(x)$ with embeddings $f_\cLevel$, $f_\monoLevel$, $f_\polyLevel$, $f_\shallowLevel$; 
	argument value $a_\cLevel$, $a_\updateLevel$, $a_\monoLevel$, $a_\polyLevel$,$a_\shallowLevel$; heap $\sigma$ and store $\mu$, we show: 
	\[\begin{array}{ll}
		\refrel{black}{R}{\cLevel}{\shallowLevel}(a_\cLevel,\sigma,a_\updateLevel,\mu,a_\monoLevel,a_\polyLevel,a_\shallowLevel,\tau,r,w)
		\land \CSem{a_\cLevel}{f_\cLevel}{\sigma}{v_\cLevel}{\sigma'} \longrightarrow \\		
		\quad \exists\ r'\ w'\ \mu'.\ r' \subseteq r \land \framerelargs{w}{\mu}{w'}{\mu'} \\
		\qquad \land\ \exists\ v_\updateLevel\ v_\monoLevel\ v_\polyLevel\ v_\shallowLevel.\\
		\qquad\quad \land\ \UpdSemAb{\xi}{(x \mapsto a_\updateLevel)}{f_\monoLevel}{\mu}{v_\updateLevel}{\mu'}\\
		\qquad\quad \land\ \ValSemA{\xi}{(x \mapsto a_\monoLevel)}{f_\monoLevel}{v_\monoLevel}\\
		\qquad\quad \land\ \ValSemA{\xi}{(x \mapsto a_\polyLevel)}{f_\polyLevel}{v_\polyLevel}\\
		\qquad\quad \land\ \SSem{a_\shallowLevel}{f_\shallowLevel}{v_\shallowLevel}\\
		\qquad\quad \land\	\refrel{black}{R}{\cLevel}{\shallowLevel}(v_\cLevel,\sigma',v_\updateLevel,\mu',v_\monoLevel,v_\polyLevel,v_\shallowLevel,\tau',r',w')
	\end{array}\]
\end{theorem}
\noindent 
Because this theorem encompasses \emph{all} levels of refinement as well as type preservation and the frame	requirements,
it is sufficient to prove this theorem about (each embedding of) each abstract function, where the 
semantics of the deep embeddings $f_\polyLevel$ and $f_\monoLevel$ is given by the user-supplied environments $\xi_v$ and $\xi_u$.
A proof of this theorem for an abstract function is sufficient to integrate its verification with the 
\Cogent verification chain.

\subsection{Requirements on Abstract Types}
\label{ssec:type-reqs}
\Cref{thm:everything} encompasses all requirements the \Cogent FFI places on abstract functions, 
but users can also provide abstract \emph{types}, extending the dynamic value typing rules as they 
see fit. Therefore, \Cogent imposes several constraints on the value typing judgements which ensure 
that key type system invariants such as memory safety are maintained for abstract types.

Consider an abstract type of the form ${``\AbsTy{A}{\overline\tau}\text{''}}$, where \AbsN{A} is the name of
the abstract type and $\overline\tau$ is a list a type parameters. Because it is not surrounded with the !-operator,
it is \emph{linear} and therefore writable. The requirements of the user-defined value typing relation are as follows.
\begin{definition}[Value Typing Requirements]\label{eq:valtype}
\begin{align*}
	\textrm{bang}_\mathsf{v}\textrm{:}\quad & \VTRV{v}{\AbsTy{A}{\overline\tau}} \longrightarrow
	\VTRV{v}{(\AbsTy{A}{\overline{\tau!}})!}\\
	\textrm{bang}_\mathsf{u}\textrm{:}\quad & \VTRU{u}{\mu}{\AbsTy{A}{\overline\tau}}{r}{w}
	\longrightarrow \VTRU{u}{\mu}{(\AbsTy{A}{\overline{\tau!}})}{r
	\cup w}{\emptyset}\\
	\textrm{no-alias:}\quad & \VTRU{u}{\mu}{\AbsTy{A}{\overline\tau}}{r}{w} \longrightarrow r
	\cap w = \emptyset\\	
	\textrm{valid:}\quad & \VTRU{u}{\mu}{\AbsTy{A}{\overline\tau}}{r}{w} \longrightarrow
	\forall p \in r \cup w.\ \mu(p) \neq \bot\\
	\textrm{frame:}\quad & \VTRU{u}{\mu_i}{\AbsTy{A}{\overline\tau}}{r}{w} \longrightarrow
	\framerelargs{w_i}{\mu_i}{w_o}{\mu_o} \\
	& \longrightarrow (r \cup w) \cap w_i =
	\emptyset \longrightarrow
	\VTRU{u}{\mu_o}{\AbsTy{A}{\overline\tau}}{r}{w}
\end{align*}
\end{definition}
\noindent The \textit{bang} rules ensure that abstract values respect the
!-operator, \ie when \CBang{} is applied to a
value, the value becomes read-only;
\textit{no-alias} ensures that there is no aliasing
of writable pointers by readable pointers within a value;
\textit{valid} enforces that all pointers are valid, \ie
point to values;
and \textit{frame} ensures that an abstract
value is unchanged if it is not part of the store that is currently being
modified. For a read-only type $(\AbsTy{A}{\overline\tau})!$, the requirements are the same, save that the writable pointer set $w$ must also be empty:
\begin{align*}
	\textrm{read-only:}\quad &
	\VTRU{u}{\mu}{(\AbsTy{A}{\overline\tau})!}{r}{w} \longrightarrow w
	= \emptyset\\	
\end{align*}

\subsection{Summary of Requirements}
As can be seen from the previous sections, proof engineers must provide a number of implementations, abstractions, and proofs for each function and type they import.
We briefly summarise these here. For each abstract type imported from C, the proof engineer simply needs to prove the requirements of \Cref{eq:valtype}. For each imported foreign
function $f$, proof engineers must define a version of the function for all semantic levels in the Cogent hierarchy. That is, they must define
a C implementation, an update semantics ($\xi_\updateLevel\ f$), a monomorphic value semantics ($\xi_\monoLevel\ f$), a polymorphic value semantics ($\xi_\polyLevel\ f$), 
and a shallow embedding in HOL. 
Once all embeddings have been provided, overall refinement (\Cref{thm:everything}) must be proven for the foreign function $f$, which necessitates proofs of all 
of our other preservation and refinement theorems:
\begin{itemize}
	\item Type preservation in the update semantics (\Cref{thm:framepres}): 
		required by all functions that call $f$ to prove type preservation (\Cref{thm:framepres}), and by $f$ and all functions that call it to prove refinement (\Cref{thm:corres,thm:updvalrefinement}).
	\item Type preservation in the monomorphic value semantics (\Cref{thm:presu}):
		required by all functions that call $f$ to prove type preservation (\Cref{thm:presu}), and by $f$ and all functions that call it to prove refinement (\Cref{thm:updvalrefinement,thm:monomorphise}).
	\item Refinement from the update semantics to C (\Cref{thm:corres}): required by all functions that call $f$ to prove refinement (\Cref{thm:corres}).
	\item Refinement from the value semantics to the update semantics (\Cref{thm:updvalrefinement}): required by all functions that call $f$ to prove refinement (\Cref{thm:updvalrefinement}).
	\item Refinement from the polymorphic to monomorphic value semantics (\Cref{thm:monomorphise}): required by all functions that call $f$ to prove refinement (\Cref{thm:monomorphise}). 
	\item Refinement from the shallow embedding to the polymorphic value semantics (\Cref{thm:scorres}): required by all functions that call $f$ to prove refinement (\Cref{thm:scorres}).
\end{itemize}

\section{Arrays}\label{sec:wordarray}
%\christine{Louis: to write initial draft}
%Word arrays are finite sequences stored on the heap, whose elements are of
%a type that has a kind $\{D, S, E\}$, i.e.\ they can left unused, shared
%freely, and be safely bound in \textbf{let}! expressions. 
%Word Arrays are arrays that are stored on the heap whose elements are primitive numeric values (primitive types are non-linear as they do not contain pointers). which are non-linear
\emph{Arrays} in \Cogent are stored on the heap
with elements having any \emph{unboxed} type. Arrays of pointers are a 
separate data structure in \cogent, as their interface is complicated 
by the uniqueness type system.
We verify five array operations:
the \walen, \waget, and \waput functions in \autoref{ex:array}, and
the \wafold and \wamap iterators in \autoref{ex:iterators}.

\subsection{Specification and Implementation}
\label{ssec:impl}
\begin{figure}
	\begin{tabular}{l}
\Walen{s} : [\CTypeVar{a}] \CFunctionArrow \CPrimType{word32}\\
\Walen{s} \CVar{xs} =
	\CFunName{of\_nat} \CLParen\CFunName{List.length} \CVar{xs}\CRParen\\
\Waget{s} : [\CTypeVar{a}] \CFunctionArrow \CPrimType{word32}
	\CFunctionArrow \CTypeVar{a}\\
\Waget{s} \CVar{xs} \CVar{i} \CVar{d} = \CKwd{if} \CFunName{unat}
	\CVar{i} < \CFunName{List.length} \CVar{xs} \\\qquad\qquad\quad\!\!\!\!\!\!\;\!  \CKwd{then} \CVar{xs} !
	\CLParen\CFunName{unat} \CVar{i}\CRParen%\\\qquad\qquad\quad 
	\CKwd{else} \CVar{d}\\
\Waput{s}: [\CTypeVar{a}] \CFunctionArrow \CPrimType{word32}
	\CFunctionArrow \CTypeVar{a} \CFunctionArrow [\CTypeVar{a}]\\
\Waput{s} \CVar{xs} \CVar{i} \CVar{v} =
	\CVar{xs}[\CFunName{unat}\ \CVar{i} $\coloneqq$ \CVar{v}]\\
\Wafold{s} : \CLParen \CTypeVar{b} \CFunctionArrow \CTypeVar{a}
	\CFunctionArrow \CTypeVar{c}
	\CFunctionArrow \CTypeVar{b} \CRParen
	\CFunctionArrow \CTypeVar{b} \CFunctionArrow [\CTypeVar{a}] \\\qquad\!\!\!\!
	\CFunctionArrow \CPrimType{word32} \CFunctionArrow \CPrimType{word32} \CFunctionArrow \CTypeVar{c}
	\CFunctionArrow \CTypeVar{b}\\
\Wafold{s} \CVar{f} \CVar{acc} \CVar{xs} \CVar{frm} \CVar{to}
	\CVar{obs} =\\\quad \CFunName{List.fold}
	\CLParen $\lambda$\CVar{acc} \CVar{el}. \CVar{f} \CVar{acc} \CVar{el}
	\CVar{obs} \CRParen \CVar{acc} \CLParen \CFunName{slice} \CVar{frm} \CVar{to} \CVar{xs} \CRParen\\
\Wamap{s} : \CLParen \CTypeVar{b} \CFunctionArrow \CTypeVar{a}
	\CFunctionArrow \CTypeVar{c}
	\CFunctionArrow \CLParen \CTypeVar{a}, \CTypeVar{b} \CRParen \CRParen
	\CFunctionArrow \CTypeVar{b} \CFunctionArrow [\CTypeVar{a}]\\\qquad\qquad\quad\!\!\!\!
	\CFunctionArrow \CPrimType{word32} \CFunctionArrow \CPrimType{word32} \CFunctionArrow \CTypeVar{c}
	\CFunctionArrow \CLParen [\CTypeVar{a}], \CTypeVar{b} \CRParen\\
\Wamap{s} \CVar{f} \CVar{acc} \CVar{xs} \CVar{frm} \CVar{to} \CVar{obs} = \\
	\quad \CKwd{let} \CLParen \CVar{xs'}, \CVar{acc'} \CRParen = 
	\CFunName{List.fold}\\\qquad
	\CLParen $\lambda$\CVar{el} \CLParen \CVar{xs}, \CVar{acc} \CRParen.
	\CKwd{let} \CLParen \CVar{el'}, \CVar{acc'} \CRParen = \CVar{f}
	\CVar{acc} \CVar{el} \CVar{obs} \\\qquad\qquad\qquad\qquad\;\,\ \!\!\!\!\CKwd{in} \CLParen \CVar{xs} @ [\CVar{el'}],
	\CVar{acc'} \CRParen \CRParen\\\qquad  \CLParen \CFunName{slice} \CVar{frm}
	\CVar{to} \CVar{xs} \CRParen \CLParen [], \CVar{acc}
	\CRParen\\\quad
	\CKwd{in} \CLParen \CFunName{take} \CVar{frm} \CVar{xs} @ \CVar{xs'} @
	\CFunName{drop} \CLParen \CFunName{max} \CVar{frm} \CVar{to} \CRParen
	\CVar{xs}, \CVar{acc'} \CRParen
\end{tabular}
	\caption{Functional correctness specification of the
	array operations in Isabelle/HOL. }
	\label{fig:spec}
\end{figure}
Our operations on arrays are specified as Isabelle functions on lists.
The specification for the array functions presented in \Cref{ex:array,ex:iterators} 
is provided in \Cref{fig:spec}. 
Most array operations have obvious analogues in Isabelle/HOL's list library. The library functions
$\mathsf{unat}$ and $\mathsf{of\dash{}nat}$ convert words and natural numbers,
$\mathsf{take}\ n\ xs$ and $\mathsf{drop}\ n\ xs$ return and remove the first $n$ elements of
the list $xs$ respectively, $\mathsf{slice}\ n\ m\ xs$ returns the sublist 
that starts at index $n$ and ends at index $m$ of the list $xs$, $!$ returns the $i^{th}$ element of a list and $@$ appends two lists.

The operations \waget and \wamap do not have straightforward analogues in the Isabelle library.
\wamap is not part of the library at all, and must be implemented in terms of \wafold, whereas
\waget behaves differently to its corresponding list operation:
our implementation of \waget is not undefined when the given index is out of bounds,
but instead returns the provided default value.

\begin{figure}
	
	\scriptsize
	\begin{tabular}{l}	
	\CKwd{struct} \CTypeName{WArray\typeparam{T}} \CLBrace\\
	\quad \CPrimType{u32} \CFieldName{len};\\
	\quad \typeparam{T}* \CFieldName{vals};\\
	\CRBrace;\\
	\CPrimType{u32}
		\CFunName{length}\CLParen\CTypeName{WArray\typeparam{T}}
		*\CVar{arr}\CRParen\CLBrace\;
		\CKwd{return} \CVar{arr}$\rightarrow$\CVar{len};\;
		\CRBrace\\
	\typeparam{T} \CFunName{get}\CLParen\CTypeName{WArray\typeparam{T}}
		*\CVar{arr}, \CPrimType{u32} \CVar{i}, \typeparam{T}\ \CVar{def}\CRParen \CLBrace\\
		\quad \CKwd{if}\CLParen\CVar{i} < \CVar{arr}$\rightarrow$\CVar{len}\CRParen
		\CLBrace\;\CKwd{return} \CVar{arr}$\rightarrow$\CVar{vals}[\CVar{i}]; \,\CRBrace\\
		\quad \CKwd{return} \CVar{def};\\
		\CRBrace\\
	\CTypeName{WArray\typeparam{T}} *\CFunName{put}\CLParen\CTypeName{WArray\typeparam{T}}
		*\CVar{arr}, \CPrimType{u32} \CVar{i}, \typeparam{T} \CVar{v}\CRParen \CLBrace\\
		\quad \CKwd{if}\CLParen\CVar{i} < \CVar{arr}$\rightarrow$\CVar{len}\CRParen
		\CLBrace\;\CVar{arr}$\rightarrow$\CVar{vals}[\CVar{i}] = \CVar{v};\,\CRBrace\\
		\quad \CKwd{return} \CVar{arr};\\
		\CRBrace\\
		\typeparam{A} \CFunName{fold}\CLParen
		\CPrimType{fid} \CVar{f},
		\typeparam{A} \CVar{acc},
		\CTypeName{WArray\typeparam{T}} *\CVar{arr},
		\CPrimType{u32} \CVar{frm},
		\CPrimType{u32} \CVar{to},
		\typeparam{O} \CVar{obsv}\CRParen \CLBrace\\
		\quad \CPrimType{u32} \CVar{i,e};\\
		\quad \CVar{e} = \CVar{arr}$\rightarrow$\CVar{len};\\
		\quad \CKwd{if}\CLParen\CVar{to} < e\CRParen
		\CLBrace\;\CVar{e} = \CVar{to};\,\CRBrace\\
		\quad \CKwd{for}\CLParen\CVar{i} = \CVar{frm}; \CVar{i} < \CVar{e};
		\CVar{i}++\CRParen\;
		\CLBrace\CVar{acc} =
		\CFunName{\textcolor{red}{dispatch\_fold}}\CLParen\CVar{f},
		\CVar{arr}$\rightarrow$\CVar{vals}[\CVar{i}], \CVar{acc},
		\CVar{obsv}\CRParen;\,\CRBrace\\
		\quad \CKwd{return} \CVar{acc};\\
		\CRBrace\\
	\CTypeName{\textcolor{gray}{ArrayAcc}} \CFunName{mapAccum}\CLParen
		\CPrimType{fid} \CVar{f},
		\typeparam{A} \CVar{acc},
		\CTypeName{WArray\typeparam{T}} *\CVar{arr},
		\CPrimType{u32} \CVar{frm},
		\CPrimType{u32} \CVar{to},
		\typeparam{O} \CVar{obsv}\CRParen \CLBrace\\
		\quad \CPrimType{u32} \CVar{i,e};\\
		\quad \CVar{e} = \CVar{arr}->\CVar{len};\\
		\quad \CKwd{if}\CLParen\CVar{to} < e\CRParen
		\CLBrace\;\CVar{e} = \CVar{to};\,\CRBrace\\
		\quad \CKwd{for}\CLParen\CVar{i} = \CVar{frm}; \CVar{i} < \CVar{e};
		\CVar{i}++\CRParen \CLBrace\\
		\qquad \CTypeName{\textcolor{gray}{ElemAcc}} \CVar{ea} =
		\CFunName{\textcolor{red}{dispatch\_mapAccum}}\CLParen\CVar{f},
		\CVar{arr}$\rightarrow$\CVar{vals}[\CVar{i}], \CVar{acc},
		\CVar{obsv}\CRParen;\\
		\qquad\CVar{arr}$\rightarrow$\CVar{vals}[\CVar{i}] = \CVar{ea}.\CVar{elem};\\
		\qquad\CVar{acc} = \CVar{ea}.\CVar{acc};\\
		\quad$ $\CRBrace\\
		\quad\CTypeName{\textcolor{gray}{ArrayAcc}} \CVar{ret} = \CLBrace.\CVar{arr} =
		\CVar{arr}, .\CVar{acc} = \CVar{acc}\CRBrace;\\
		\quad\CKwd{return} \CVar{ret};\\
		\CRBrace\\
	\end{tabular}
	\caption{C implementation of key operations on arrays.}
	\label{fig:impl}
\end{figure}

In \autoref{fig:impl}, we present a version of the template C
implementation of arrays with some syntactic simplifications for presentation. Quoted type parameters that are later instantiated
to concrete \Cogent types, such as \typeparam{T} or \typeparam{O}, are highlighted in blue.
Generated C structure types for Cogent records and tuples such as \textcolor{gray}{$\mathit{ArrayAcc}$}, which normally have compiler-generated 
names, are written with human-friendly names in grey. The \emph{dispatch} functions highlighted in \textcolor{red}{red}
are how \Cogent deals with higher-order functions: Because the C-parser semantics do not accommodate 
function pointers, \Cogent instead assigns a unique identifier to each function at compile time, and defines
a dispatch function for each function type. This function takes a function identifier as an
argument and calls the function that corresponds to that identifier.

\subsection{Proving Refinement}
As previously mentioned, to verify our abstract C library, we must additionally provide Isabelle/HOL abstractions of our C implementation 
that can connect with the automatically-generated \Cogent embeddings. Specifically, we must provide 
Isabelle/HOL abstractions for each abstract function (\ie \waput, \waget, etc.), as well as a combined 
value-update-type correspondence relation, analogous to the refinement relation for~\Cref{thm:updvalrefinement},
 for each abstract type (\ie \CTypeName{WArray\typeparam{T}}).

 \subsubsection{Abstractions}
\begin{figure*}
	\begin{subfigure}[t]{0.50\textwidth}
\begin{tabular}{l}
\Walen{u}\CLParen$\mu$, \CVar{x}\CRParen \CLParen$\mu$',
	\CVar{y}\CRParen = \\\quad $\exists$\CVar{\tau} \CVar{len} \CVar{p}. \CLParen
	\CVar{\mu} \CVar{x} = \CConstructor{UWA} \CVar{\tau} \CVar{len} \CVar{p} \CRParen
	$\land$ \CLParen \CVar{y} = \CVar{len} \CRParen $\land$ \CLParen
	$\mu$ = $\mu$' \CRParen\\
\Waget{u}\CLParen\CVar{\mu}, (\CVar{x},
	\CVar{i}, \CVar{d})\CRParen \CLParen\CVar{\mu'}, \CVar{y}\CRParen =\\
\quad 	$\exists$\CVar{\tau} \CVar{len} \CVar{p}. \CLParen
	\CVar{\mu} \CVar{x} = \CConstructor{UWA} \CVar{\tau} \CVar{len}
	\CVar{p}\CRParen\\\qquad $\land$ \CLParen\CVar{i} < \CVar{len}
	$\longrightarrow$
	\CVar{\mu} (\CVar{p}+(\CFunName{size} \CVar{t})$\times$\CVar{i}) = \CVar{y}\CRParen
	\\\qquad$\land$ \CLParen\CVar{i} $\geq$ \CVar{len} $\longrightarrow$ \CVar{y} =
	\CVar{d}\CRParen $\land$ \CLParen\CVar{\mu} = \CVar{\mu'}\CRParen\\
\Waput{u}\CLParen\CVar{\mu}, (\CVar{x}, \CVar{i}, \CVar{v})\CRParen
	\CLParen\CVar{\mu'}, \CVar{y}\CRParen =\\
	\quad $\exists$\CVar{\tau} \CVar{len} \CVar{p}.
	\CLParen\CVar{\mu} \CVar{x} = \CConstructor{UWA} \CVar{\tau} \CVar{len}
	\CVar{p}\CRParen \\ \qquad $\land$ \CLParen\CVar{i} < \CVar{len}
	$\longrightarrow$
	\CVar{\mu}(\CVar{p}+(\CFunName{size} \CVar{\tau})$\times$\CVar{i} $\mapsto$ \CVar{v}) = \CVar{\mu'}\CRParen
	\\\qquad $\land$ \CLParen\CVar{i} $\geq$ \CVar{len} $\longrightarrow$
	\CVar{\mu} = \CVar{\mu'}\CRParen $\land$ \CLParen\CVar{x} =
	\CVar{y}\CRParen\\
\Wafold{u} \CLParen\CVar{\mu_1},
	(\CVar{f}, \CVar{acc}, \CVar{x}, \CVar{s}, \CVar{e}, \CVar{obs})\CRParen
	\CLParen\CVar{\mu_3}, \CVar{y}\CRParen =\\\quad $\exists$\CVar{\tau} \CVar{len} \CVar{p}.
	\CLParen\CVar{\mu_1} \CVar{x}  = \CConstructor{UWA} \CVar{\tau} \CVar{len}
	\CVar{p}\CRParen \\\qquad $\land$ \CLParen\CVar{s} < \CVar{len} $ \land$
	\CVar{s} < \CVar{e}
	$\longrightarrow$\\\qquad\quad
	$\exists$\CVar{v} \CVar{acc'} \CVar{\mu_2}.
	\CVar{\mu_1} (\CVar{p}+(\CFunName{size} \CVar{\tau})$\times$\CVar{s}) =
	\CVar{v} \\\qquad\qquad $\land$
	$\UpdSemAb{\xi_u}{[a \mapsto (\text{\CVar{v}}, \mathit{acc},
	\mathit{obs})]}{f a\ }{\ \mu_1}{\mathit{acc'}\ }{\ \mu_2}$\\\qquad\qquad $\land$
	\Wafold{u} \CLParen\CVar{\mu_2},
	(\CVar{f}, \CVar{acc'}, \CVar{x}, \CVar{s}+1, \CVar{e}, \CVar{obs})\CRParen
	\CLParen\CVar{\mu_3}, \CVar{y}\CRParen\CRParen
	\\\qquad $\land$ \CLParen\CVar{s} $\geq$ \CVar{len} $\lor$ \CVar{s} $\geq$
	\CVar{e}$\longrightarrow$
	\CVar{\mu_1} = \CVar{\mu_3} $\land$ \CVar{y} = \CVar{acc}\CRParen\\
\Wamap{u} \CLParen\CVar{\mu_1},
	(\CVar{f}, \CVar{acc}, \CVar{x}, \CVar{s}, \CVar{e}, \CVar{obs})\CRParen
	\CLParen\CVar{\mu_4}, \CVar{y}\CRParen =
	\\\quad $\exists$\CVar{\tau} \CVar{len} \CVar{p}.
	\CLParen\CVar{\mu_1} \CVar{x} = \CConstructor{UWA} \CVar{\tau} \CVar{len}
	\CVar{p}\CRParen \\\qquad $\land$ \CLParen\CVar{s} < \CVar{len} $ \land$
	\CVar{s} < \CVar{e}
	$\longrightarrow$
	\\\qquad\quad$\exists$\CVar{v} \CVar{v'} \CVar{acc'} \CVar{\mu_2}
	\CVar{\mu_3}.
	\CVar{\mu_1} (\CVar{p}+(\CFunName{size} \CVar{\tau})$\times$\CVar{s}) =
	\CVar{v} \\\qquad\qquad$\land$
	$\UpdSemAb{\xi_u}{[a \mapsto (\text{\CVar{v}}, \mathit{acc},
	\mathit{obs})]}{f a\ }{\ \mu_1}{(\text{\CVar{v'}}, \mathit{acc'})\ }{\ \mu_2}$ \\\qquad\qquad$\land$
	\CVar{\mu_3} = \CVar{\mu_2}(\CVar{p}+(\CFunName{size} \CVar{t})$\times$\CVar{s}
	$\mapsto$ \CVar{v'}) \\\qquad\qquad $\land$
	\Wamap{u} \CLParen\CVar{\mu_3},
	(\CVar{f}, \CVar{acc'}, \CVar{x}, \CVar{s}+1, \CVar{e}, \CVar{obs})\CRParen
	 \CLParen\CVar{\mu_4}, \CVar{y}\CRParen\CRParen\\\qquad
	$\land$ \CLParen\CVar{s} $\geq$ \CVar{len} $\lor$ \CVar{s} $\geq$
	\CVar{e}$\longrightarrow$
	\CVar{\mu_1} = \CVar{\mu_4} $\land$ \CVar{y} = (\CVar{x}, \CVar{acc})\CRParen
\end{tabular}
		\caption{C abstractions for the \textit{update} semantics.}
		\label{fig:uembed}
	\end{subfigure}\hfill
	\begin{subfigure}[t]{0.475\textwidth}
\begin{tabular}{l}
\Walen{v}\CVar{x} \CVar{y} = \\\quad$\exists$\CVar{\tau} \CVar{xs}. \CLParen
	\CVar{x} = \CConstructor{VWA} \CVar{\tau} \CVar{xs} \CRParen
	$\land$ \CLParen \CFunName{List.length} \CVar{xs} =
	\CFunName{unat} \CVar{y} \CRParen\\
\Waget{v}\CLParen\CVar{x}, \CVar{i}\CRParen \CVar{y} =\\
	\quad $\exists$\CVar{\tau} \CVar{xs}. \CLParen
	\CVar{x} = \CConstructor{VWA} \CVar{\tau} \CVar{xs} \CRParen \\\qquad $\land$
	\CLParen\CFunName{unat} \CVar{i} < \CFunName{List.length} \CVar{xs}
	$\longrightarrow$ \CVar{xs} ! \CVar{i} = \CVar{y} \CRParen \\\qquad $\land$
	\CLParen\CFunName{unat} \CVar{i} $\geq$
	\CFunName{List.length} \CVar{xs} $\longrightarrow$ \CVar{y} =
	0\CRParen\\
\Waput{v}\CLParen\CVar{x}, \CVar{i}, \CVar{v}\CRParen \CVar{y} =
	\\\quad $\exists$\CVar{\tau} \CVar{xs}. \CLParen
	\CVar{x} = \CConstructor{VWA} \CVar{\tau} \CVar{xs} \CRParen  $\land$
	\CLParen\CVar{y} = \CConstructor{VWA} \CVar{\tau} \CVar{xs}[\CFunName{unat} \CVar{i}
	$\coloneqq$ \CVar{v}]\CRParen\\
\Wafold{v}
	\CLParen\CVar{f}, \CVar{acc}, \CVar{x}, \CVar{s}, \CVar{e}, \CVar{obs}\CRParen
	\CVar{y} =\\\quad
	 $\exists$\CVar{\tau} \CVar{xs}. \CLParen
	\CVar{x} = \CConstructor{VWA} \CVar{\tau} \CVar{xs} \CRParen \\\qquad $\land$
	\CLParen\CFunName{unat} \CVar{s} < \CFunName{List.length} \CVar{xs}
	$\land$ \CVar{s} < \CVar{e} $\longrightarrow$ 
	\\\qquad\quad $\exists$\CVar{v} \CVar{acc'}. (\CVar{xs} ! \CVar{s} = \CVar{v})
	\\\qquad\qquad$\land$
	$\ValSemA{\xi_v}{[a \mapsto (\mathit{v}, \mathit{acc}, \mathit{obs})]}{f
	a}{\mathit{acc'}}$ 
	\\\qquad\qquad$\land\ $\Wafold{v}
	\CLParen\CVar{f}, \CVar{acc'}, \CVar{x}, \CVar{s}+1, \CVar{e}, \CVar{obs}\CRParen
	\CVar{y}\CRParen \\\qquad$\land$
	\CLParen\CFunName{unat} \CVar{s}  $\geq$
	\CFunName{List.length} \CVar{xs} $\lor$ \CVar{s} $\geq$ \CVar{e}
	$\longrightarrow$ \CVar{y} = \CVar{acc}\CRParen\\
	\Wamap{v}
	\CLParen\CVar{f}, \CVar{acc}, \CVar{x}, \CVar{s}, \CVar{e}, \CVar{obs}\CRParen
	\CVar{y} =
	\\\quad$\exists$\CVar{\tau} \CVar{xs}. \CLParen
	\CVar{x} = \CConstructor{VWA} \CVar{\tau} \CVar{xs} \CRParen\\\qquad $\land$
	\CLParen\CFunName{unat} \CVar{s} < \CFunName{List.length} \CVar{xs}
	$\land$ \CVar{s} < \CVar{e} $\longrightarrow$
	\\\qquad\quad$\exists$\CVar{v} \CVar{v'} \CVar{acc'} \CVar{x'}.
	\CVar{xs} ! \CVar{s} = \CVar{v}	\\\qquad\qquad$\land$
	$\ValSemA{\xi_v}{[a \mapsto (\mathit{v}, \mathit{acc}, \mathit{obs})]}{f
	a}{(v', \mathit{acc'})}$ \\\qquad\qquad $\land$
	\CVar{x'} = \CConstructor{VWA} \CVar{\tau}
	\CVar{xs}[\CFunName{unat} \CVar{s} $\coloneqq$ \CVar{v'}] \\
    \qquad\qquad$\land\ $\Wamap{v}
	\CLParen\CVar{f}, \CVar{acc'}, \CVar{x'}, \CVar{s}+1, \CVar{e}, \CVar{obs}\CRParen
	\CVar{y}\CRParen \\\qquad $\land$
	\CLParen\CFunName{unat} \CVar{s} $\geq$
	\CFunName{List.length} \CVar{xs} $\lor$ \CVar{s} $\geq$ \CVar{e}
	$\longrightarrow$ \CVar{y} = (\CVar{x}, \CVar{acc})\CRParen
\end{tabular}
		\caption{C abstractions for the \textit{value} semantics.}
		\label{fig:vembed}
	\end{subfigure}
	\caption{\Cogent-compatible abstractions of C operations on arrays}
	\label{fig:embedding}
\end{figure*}

We extend the definition of \Cogent values with a new constructor for arrays.
In the update semantics, arrays take the form $\text{\CConstructor{UWA}}\ \tau\ \mathit{len}\ \mathit{p}$, where
$\tau$ is the type of the array elements, $\mathit{len}$ is the length of the array, and
$\mathit{p}$ is a pointer to the first element of the array. This is quite similar to 
the representation used in C, where a struct containing the length and a pointer is used (see~\Cref{fig:impl}). 
We additionally store the type in the \Cogent value so that the same form of \Cogent value can be 
used for all the different structs generated by the \Cogent compiler from the C template.
In the value semantics, however, arrays take the form $\text{\CConstructor{VWA}}\ \tau\ \mathit{xs}$, where $\mathit{xs}$ is simply 
an Isabelle/HOL list of Cogent values, similar to the representation used in the shallow embedding (see~\Cref{fig:spec}).

The various array function
abstractions supplied to \Cogent are defined as input-output relations that are later interpreted as functions, as shown in
\Cref{fig:embedding}. These abstractions are derived directly from the axiomatisation
used in the verification of \bilby.
In the update semantics (\Cref{fig:uembed}), the input is a pair of the \Cogent store before the array operation and the argument(s) to the
function, and the output is a pair of the store after the operation is complete and the return value.

Note that the definitions for \wafold and \wamap are recursive in two ways. The most obvious is the direct structural recursion on the array indices,
which is straightforward for Isabelle to show terminates. The other form of recursion is the mutual recursion with 
the \Cogent semantics (note that \wafold and \wamap invoke the \Cogent semantics to evaluate the argument function).
This is because the \Cogent semantics is itself parameterised 
by an environment ($\xi_u$ and $\xi_v$) containing the abstractions we are currently defining.
We cannot define these embeddings and the \Cogent semantics simultaneously as is normally done for mutual definitions, however, since our abstractions are
defined by users later in the process after the semantics of \Cogent have already been defined. Thus, we cannot close the recursion before
users have provided these definitions.
Therefore we must impose an additional constraint\footnote{This constraint is not a proof obligation, but merely a requirement for our framework to work automatically.} that an $n$-order abstract function can only
call a function of an order less than $n$.
This forces us to prove refinement for all functions of order $n-1$ or
less, before we can prove refinement for an abstract function of order $n$. By imposing this ordering, we can iteratively 
build up these environments in a staged way.
This constraint is not burdensome, and is already satisfied by all \Cogent codebases (see Rizkallah et al.~\citep{Rizkallah_LNSCOMKK_16}). 

Since \Cogent supports polymorphism, we in fact need to provide two abstractions for
our operations in the value semantics corresponding to the two value semantics layers
in our refinement chain, monomorphic and polymorphic. Because our abstractions (and even our implementation)
are entirely parametric in the element type, the abstractions
that are shown in \Cref{fig:vembed} can be used for both the monomorphic
and the polymorphic layers, which, as we will see, trivialises the refinement proof for the monomorphisation 
layer of the hierarchy.

\subsubsection{Value Typing}
Unlike the approach for native \Cogent values, where the value
typing relations are defined as erasures of the value-update refinement relation,
we shall define individual value-typing relations for arrays in the two semantics,
and then combine them into a refinement relation in~\Cref{def:uvvalrel}

We define these typing relations with two equations, one for the 
writable arrays $\mathtt{Array}\ \tau$, and one for read-only arrays $(\mathtt{Array}\ \tau)!$. 
In the value semantics, these two types are identical, merely requiring that the list elements are 
well-typed:
\begin{definition}[Array: Value Semantics Value Typing]\label{def:vvaltype}
\[\begin{array}{lcl}
	\VTRV{\mathtt{VWA}\ \tau\ \mathit{xs}}{\mathtt{Array}\ \tau}
	\!\!\!\!&\equiv&\!\!\!\!(\forall i < \mathsf{List.length}\ \mathit{xs}.\ \VTRV{\mathit{xs}\ !\ i}{\tau})\\
	\VTRV{\mathtt{VWA}\ \tau\ \mathit{xs}}{(\mathtt{Array}\ \tau)!}
	\!\!\!\!&\equiv&\!\!\!\!(\VTRV{(\mathtt{VWA}\ \tau\ \mathit{xs})}{\mathtt{Array}\ \tau})
\end{array}\]
\end{definition}

For the update semantics, we define an auxiliary predicate $\mathit{okay}(\mathtt{UWA}\ \tau\ \mathit{len}\ p)$ which states that 
each of the values in the array (located at successive pointers starting at $p$) is well typed, as well as a necessary condition on $\mathit{len}$ to ensure that 
our pointer arithmetic will not result in overflow. This predicate is used in the typing relation for both the 
writable and read-only array types. The heap footprint $[ r \ast w ]$ must consist of not just the pointer $p$ but all of the 
successive pointers to each array element, because all of these memory locations are contained in the array, and ownership of all of them 
is passed along with the array. Because the array contains only unboxed types, we know that there are no other pointers in the heap footprint.
For the read-only array type, $r$ contains the heap footprint and $w$ is empty, and vice versa for the writable array type.
\begin{definition}[Array: Update Semantics Value Typing]\label{def:uvaltype}
\[\begin{array}{l}
	\mathit{okay}\ (\mathtt{UWA}\ \tau\ \mathit{len}\ p) \equiv (\mathsf{unat}\ \mathit{len}\times \mathsf{size}\ \tau\leq\mathtt{max\_word})\\
	\land\ (\forall i < \mathit{len}.\ \exists v.\
		\mu\ (p+\mathsf{size}\ \tau\times i)= v\ \land\ 
		\VTRU{v}{\mu}{\tau}{\emptyset}{\emptyset})\\[0.25em]
	\VTRU{\mathtt{UWA}\ \tau\ \mathit{len}\
	p\ }{\mu}{(\mathtt{Array}\ \tau)!}{r}{w} \equiv \mathit{okay}\ (\mathtt{UWA}\ \tau\ \mathit{len}\ p) \\
	\qquad\qquad\qquad\quad\quad\land\ ( r = 
	\sset{p+i\ |\ \forall i.\ i<\mathit{len}} \land w = \emptyset)\\[0.1em]
	\VTRU{\mathtt{UWA}\ \tau\ \mathit{len}\
	p\ }{\mu}{(\mathtt{Array}\ \tau)}{r}{w} \equiv \mathit{okay}\ (\mathtt{UWA}\ \tau\ \mathit{len}\ p) \\
	\qquad\qquad\qquad\qquad\land\ (w = 
	\sset{p+i\ |\ \forall i.\ i<\mathit{len}} \land r = \emptyset)
\end{array}\]
where $\mathtt{max\_word}$ is $2^{32} - 1$ (as we are using 32-bit pointers).
\end{definition}
\noindent Recall that these the value typing relations for abstract types must
satisfy the constraints of \Cref{eq:valtype}. Because arrays
only have elements which are of unboxed types, these constraints are 
trivial to discharge.

\subsubsection{Refinement Relations}
\label{sec:refrels}
Just as \Cogent's typing relations and semantics are extended by our rules for arrays,
so too are the various refinement relations in each layer of the semantics.
Because, as previously mentioned, the C implementation of our arrays bears a strong resemblance 
to our update semantics values, and our value semantics values bear a strong resemblance to our 
Isabelle/HOL list representation, the value relations \valrelupdC and \shallowrel
for arrays are very simple.
 In the former, the two a related
if the length values and the pointer values are equal. 
In the latter, the two are related if the length of the lists
are the same and the elements are pairwise related.
\begin{definition}[\texttt{U32} Array: Update $\Rightarrow$ C Value Relation]\label{def:cvalrel}
	\[		\valrelupdC(x_\cLevel,{\mathtt{UWA}\ \mathtt{U32}\ \mathit{len}_\updateLevel\ p_\updateLevel}) \equiv (\mathit{len}_\updateLevel = \mathtt{len_\cLevel}\ x_\cLevel\ \land\ p_\updateLevel = \mathtt{arr_\cLevel}\ x_\cLevel)
	\]
where $\mathtt{len_\cLevel}$ and $\mathtt{arr_\cLevel}$ are the struct projections generated by 
AutoCorres for the array type.
\end{definition}

\begin{definition}[Array: Shallow $\Rightarrow$ Polymorphic Relation]\label{def:svalrel}
	\[\begin{array}{l}\shallowrel(x_\shallowLevel,\mathtt{VWA}\ \tau\ \mathit{xs}_\polyLevel)
	\equiv (\mathsf{length}\ \mathit{xs}_\polyLevel = \mathsf{length}\ x_\shallowLevel)\\\qquad \land\ (\forall
		i < \mathsf{length}\ \mathit{xs}_\polyLevel.\ \shallowrel(x_\shallowLevel\ !\ i,\mathit{xs}_\polyLevel\ !\
		i))
	\end{array}\]
\end{definition}

\noindent Because we use the same abstractions for both monomorphic and polymorphic layers,
the refinement relation \monorel is just equality and its refinement 
theorem is trivial. 

Lastly, it remains to define the value relation for arrays between the two \Cogent semantics, update and value.
As previously mentioned, we make use of the two value typing relations here with additional conditions to pairwise relate the corresponding elements 
of the two values:
\begin{definition}[Array: Value (Mono) $\Rightarrow$ Update Relation]
	\label{def:uvvalrel}
	\[\begin{array}{l}
		\VTRN{\mathtt{UWA}\ t\ \mathit{len}_\updateLevel\ p_\updateLevel\ }{\ \mu}{\mathtt{VWA}\ t\
		\mathit{xs}_\monoLevel}{\tau}{r}{w} 
		\equiv\\\quad\quad (\mathsf{unat}\ \mathit{len}_\updateLevel = \mathsf{length}\ \mathit{xs}_\monoLevel)\
		\\
		\quad \land\ (\forall i < \mathit{len}_\updateLevel.\ \exists \mathit{v}_\updateLevel.\
		\mu(p_\updateLevel+(\mathsf{size}\ t)\times i) = \mathit{v}_\updateLevel\\\qquad\qquad \land\ \VTRN{\mathit{v}_\updateLevel\ }{\ \mu}{(\mathit{xs}_\monoLevel\
		!\ \mathsf{unat}\ i)}{t}{\emptyset}{\emptyset})\\\quad \land\ 
		(\VTRV{\mathtt{VWA}\ t\ \mathit{xs}_\monoLevel}{\tau})\ \land\
		(\VTRU{\mathtt{UWA}\ t\
		\mathit{len}_\updateLevel\ p_\updateLevel\ }{\ \mu}{\tau}{r}{w})
	\end{array}\]
\end{definition}
\noindent Note the heap footprints of elements are always empty, as the array can only contain unboxed values.
\subsubsection{Refinement}
Now that we have all our abstractions, value typing and refinement relations, we have all the
ingredients we need to prove refinement for our array operations.

The theorems structurally resemble the refinement theorems presented in~\Cref{ssec:cogent-verification}.
For first order functions \walen, \waget and \waput , the
proofs tend to
follow easily from the definition of their abstractions, implementation, refinement relation 
and value typing relation --- the creativity is largely in the definitions, not the proofs. 
This is because we want the proofs to be easily automatable in future. Nonetheless
 we shall sketch the proofs for our \waput operation, specifically for arrays of \texttt{U32}, as an illustrative example.

We first show that the update semantics abstraction is
refined by the embedding of the C implementation that is automatically
generated by AutoCorres.  This appears similar to \Cref{thm:corres}, but instead of invoking the \Cogent update semantics we instead appeal to our abstract function from the environment
$\xi_\updateLevel(\mathsf{put}_\texttt{U32})$:
\begin{theorem}[Verifying \waput: Update $\Rightarrow$ C refinement]\label{thm:cput}\ \\
	Where the C embedding of $\mathsf{put}_\mathsf{U32}$ is $\mathit{put}_\cLevel$:
	\[\begin{array}{ll}
	\valrelupdC(a_\cLevel,a_\updateLevel) \land \heaprelupdC(\sigma,\mu)
		\land \CSem{a_\cLevel}{\sigma}{\mathit{put}_\cLevel}{r_\cLevel}{\sigma'} \longrightarrow\\
		\qquad \exists\ \mu'\ r_\updateLevel.\  \xi_u(\mathsf{put}_\texttt{U32})(\mu,a_\updateLevel) = (\mu',r_\updateLevel) \\ 
		\qquad\qquad \land\ \valrelupdC(r_\cLevel,r_\updateLevel) \land \heaprelupdC(\sigma',\mu') 
	\end{array}\]
\end{theorem}
\begin{proof}
Recall that the argument to \waput for arrays of 32-bit words is a tuple that
contains an index, the array, and a 32-bit word to write to it.
We take cases on the index. 
In the case that the index is out of bounds, both the abstraction
(\autoref{fig:uembed}) and the implementation (\Cref{fig:impl}) return the array
unmodified with the store unmodified as well, and so the theorem is trivial.
In the case where the index is within bounds, our value relation \valrelupdC on the arguments 
implies that the two argument words and indices are the same. 
Writing the same value to the same index within bounds only writes corresponding 
values to the same store locations,
so it follows that our heap relation \heaprelupdC is preserved.
As it is destructively updated, the actual location of the array in memory is not changed, so 
the relation \valrelupdC is trivially preserved to the output array.
\end{proof}
\noindent For the next level up in the refinement hierarchy, we must show refinement from 
our value semantics abstraction (\Cref{fig:vembed}) to our update semantics abstraction (\Cref{fig:uembed}).
Even though this refinement step is below monomorphisation in our hierarchy, our
abstractions for \Waput{} are agnostic to the element type of
the array, so we can generalise the proof to arrays of any
element type. The theorem resembles~\Cref{thm:updvalrefinement}, but with 
the \Cogent semantics replaced with our supplied abstractions in $\xi_\mathsf{v}$ and $\xi_\mathsf{u}$.
\begin{theorem}[Verifying \waput: Value $\Rightarrow$ Update refinement]
\label{thm:uvput} For an element type $t$,
if $\ \VTRN{a_\updateLevel}{\mu}{a_\monoLevel}{(\mathtt{Array}\
	t,\mathtt{U32}, t)}{r}{w}$ and
	$\xi_\updateLevel(\mathsf{put}_t)(\mu,a_\updateLevel)=(\mu',v_\updateLevel)$, then there exists a value $v_\monoLevel$ and pointer sets $r' \subseteq r$ and  $w'$ such that $\xi_\mathsf{v}(\mathsf{put}_t)(a_\monoLevel) = v_\monoLevel$, and $\ \VTRN{v_\updateLevel}{\mu'}{v_\monoLevel}{\mathtt{Array}\ t}{r'}{w'}$ and $\Frame{w}{\mu}{w'}{\mu'}$.	
\end{theorem}
\begin{proof}	
We also prove this by cases on
the index. In the case where the index is out of bounds, we trivially have
correspondence. In
the case where the index is within bounds, we prove that modifying the
element at the given index from the pointer on the store is equivalent to modifying the
element at the given index in the corresponding list. To prove this, we
need to show that there is a one-to-one mapping between store addresses of elements to
list indices. This is why we include in our typing relation that the array element addresses 
do not overflow the heap (\Cref{def:uvaltype}). 
Since this tells us the array cannot wrap around itself, each
element in the array has a unique address, giving us our mapping. We also
need to show that that the frame
conditions are satisfied, but because our implementation is memory safe, these follow easily from 
our definitions.
\end{proof}
\noindent Next, we must show refinement from the polymorphic layer to the monomorphic layer, but because \waput is a first-order function (neither taking functions as arguments nor returning them),
the value relations for its arguments and return values simplify to equality. Furthermore, as our \Cogent 
abstractions for monomorphic and polymorphic layers are identical, the proof of refinement is trivialised to
showing that identical functions will give equal results given equal inputs.

Monomorphisation thus easily dispatched, we must now make the final shift to the specification level (\Cref{fig:spec}).
We must prove a theorem analogous to~\Cref{thm:scorres}. Note that this theorem is also generic for any element type:
\begin{theorem}[Verifying \waput{}: Shallow $\Rightarrow$ Polymorphic Value]\label{thm:putscorres}
	The shallow embedding of \waput is called $\mathsf{put}_\shallowLevel$.
	Given arguments $a_\shallowLevel$ and $a_\polyLevel$, we have:
		\[\begin{array}{ll}
		\shallowrel(a_\polyLevel,a_\shallowLevel)
			\land \xi_\mathsf{v}(\mathsf{put})(a_\polyLevel) = r_\polyLevel \longrightarrow\\
			\qquad\qquad\qquad\;\exists r_\shallowLevel.\ \SSem{a_\shallowLevel}{\mathsf{put}_\shallowLevel}{r_\shallowLevel} \land \shallowrel(r_\polyLevel,r_\shallowLevel)
		\end{array}\]	
\end{theorem}
\begin{proof}
Follows from the definitions of \Waput{v} (in $\xi_\mathsf{v}$) and \Waput{s}, as well as the value relation 
from \Cref{def:svalrel}.
\end{proof}
\noindent The above theorems are all we need to compose the correctness of \waput with 
\Cogent's refinement hierarchy. For higher-order functions \wafold and \wamap the proofs are broadly similar, but 
slightly complicated by the presence of functions as arguments. Our theorems assume 
that type preservation and refinement hold for all of their argument functions --- an assumption that is discharged 
by the Cogent compiler (for Cogent functions) or by manual proofs (for C functions).
As always with looping functions, the majority of the proof effort was
concentrated on proving that loop invariants are maintained, and not on any aspect of the \Cogent framework.
\section{Generic Loops}
\label{sec:loops}

Aside from arrays and their associated iterators and operations, the most commonly used library functions in
\bilby are generic loop functions. These are used to write loops that do not simply iterate over a particular data structure,
and are often used to accommodate non-standard search patterns over data structures. For example, in~\Cref{sec:binsearch} we 
shall use such a function in a binary search. 

We shall verify the function \CFunName{repeat}, which is given the following type signature in Cogent:
\[\text{\CFunName{repeat} : \CLParen \CPrimType{U32},
			(\CTypeVar{a}, \CTypeVar{b}\CBang) \CFunctionArrow \CPrimType{Bool},  
			(\CTypeVar{a}, \CTypeVar{b}\CBang) \CFunctionArrow \CTypeVar{a},  
			\CTypeVar{a}, \CTypeVar{b}\CBang\CRParen
			\CFunctionArrow \CTypeVar{a}}\]
The expression $\text{\CFunName{repeat}}\ \mathit{n}\ \mathit{stop}\ \mathit{step}\ \mathit{acc}\ \mathit{obs}$ operates on some mutable state $\mathit{acc}$ (of linear type) and some observer data $\mathit{obs}$ (of read-only type), 
and runs the loop body $\mathit{step}$ on it at most $n$ times, or until $\mathit{stop}$ returns true.
\Cref{fig:loopspec} gives the Isabelle/HOL shallow embedding and \Cref{fig:loopimpl} gives the template C implementation. \Cref{fig:embeddingloop} gives the Cogent-compatible 
embeddings for the environments $\xi_{u}$ and $\xi_\mathsf{v}$. Note that these must invoke the Cogent semantics twice, once to evaluate each of the argument functions.
Discharging the required proof obligations connecting all of these embeddings and maintaining \Cogent's invariants is even more straightforward than for 
other higher-order functions such as \textsf{mapAccum}, as here we do not even need to consider custom data structures such as arrays. 
This means that we can even define a polymorphic abstraction of the C code compatible with AutoCorres, enabling us to prove the entire refinement chain polymorphically and thereby largely eliminate the boilerplate of multiple type instantiations.

\begin{figure}
	\begin{tabular}{l}
		\RepeatAtm{s} : \CPrimType{nat} \CFunctionArrow 
			\CLParen\CTypeVar{a} \CFunctionArrow \CTypeVar{b} \CFunctionArrow \CPrimType{bool}\CRParen \CFunctionArrow
			\CLParen\CTypeVar{a} \CFunctionArrow \CTypeVar{b} \CFunctionArrow \CTypeVar{a}\CRParen\\
			\qquad\quad\!\!\!\!\!
			\CFunctionArrow \CTypeVar{a} \CFunctionArrow \CTypeVar{b}
			\CFunctionArrow \CTypeVar{a}\\
		\RepeatAtm{s} 0 \_ \_ \CVar{acc} \_ = \CVar{acc}\\
		\RepeatAtm{s} \CLParen\CFunName{Suc} \CVar{n}\CRParen \CVar{f} \CVar{g} \CVar{acc} \CVar{obsv} = 
		 \CKwd{if} \CLParen\CVar{f} \CVar{acc} \CVar{obsv} \CRParen \\\quad \CKwd{then} \CVar{acc}
		 \CKwd{else} \RepeatAtm{s} \CVar{n} \CVar{f} \CVar{g} \CLParen\CVar{g} \CVar{acc} \CVar{obsv}\CRParen \CVar{obsv}
	\end{tabular}
	\caption{Shallow embedding of the generic loop}
	\label{fig:loopspec}
\end{figure}

\section{Composing Verification of \Cogent and C}
\label{sec:examples}
Now that we have demonstrated how to prove the obligations placed on C code, we shall illustrate 
how to integrate these proofs with proofs about \Cogent code.
Firstly, as a simple example, we return to the \texttt{sum} function presented in \Cref{ex:iterators}.

After compiling this program and linking it to our C library, the compiler produces a refinement theorem similar to 
\Cref{thm:everything}, however the phases that are generated per-program via translation validation such as the final refinement to C
leave open proof obligations for the user to discharge about abstract functions\footnote{At the time of writing, Cogent's shallow phase (\Cref{thm:scorres}) implicitly assumes abstract function correctness rather than doing so explicitly as done in other phases. So for now, we just copy and discharge these.}. In our sum example, \Cogent will generate obligations 
about \CFunName{length} and about \CFunName{fold}. The obligation about \CFunName{length} is exactly our C refinement theorem for \CFunName{length} (the \CFunName{length} analogue of \Cref{thm:cput}), and the obligation for \CFunName{fold} is an instance of our theorem for \CFunName{fold}: the obligation requires showing refinement under the assumption that the argument function is \CFunName{add}, whereas our theorem is generically proven for any function that maintains \Cogent's invariants and refinement.
Because \CFunName{add} is defined in \Cogent, \Cogent generates the required theorem for the argument function for us, allowing us to easily discharge this obligation.

We additionally must instantiate our sets of abstract values and types with arrays, and the environments $\xi_\textsf{v}$ and $\xi_\textsf{u}$ with our abstractions from \Cref{fig:embedding}. This instantiation 
requires us to additionally provide proofs similar to that of \Cref{thm:uvput} for each function, as well as proofs that all the type conditions from~\Cref{ssec:type-reqs}. These proofs ultimately connect our C code to
the generated shallow embedding, which strongly resembles the original \Cogent code:
$$
\begin{array}{l}
\mathsf{add}_\shallowLevel\ x\ y = x + y\\
\mathsf{sum}_\shallowLevel\ \mathit{xs} = \mathsf{fold}_\shallowLevel\ \mathsf{add}_\shallowLevel\ 0\ \mathit{xs}\ 0\ (\mathsf{length}_s\ \mathit{xs})\ ()\\
\end{array}
$$
With all of these proofs in place, we get a refinement theorem like~\Cref{thm:everything} for the function $\mathsf{sum}$, which leaves no function unverified. This refinement theorem 
allows us to prove properties of $\mathsf{sum}_\shallowLevel$ just by equational reasoning, and have these proofs also apply to the C implementation.
%\christine{present final theorem}
\begin{figure}
	
	\begin{tabular}{l}
		\typeparam{A} \CFunName{repeat}\CLParen \CPrimType{u32} \CVar{n},
			\CPrimType{fid} \CVar{f}, \CPrimType{fd} \CVar{g},
			\typeparam{A} \CVar{acc}, \typeparam{O} \CVar{obsv}\CRParen \CLBrace\\
		\quad \CKwd{for}\CLParen\CPrimType{u32} \CVar{i} = 0; \CVar{i} < \CVar{n}; \CVar{i}++\CRParen\\
			\qquad \CKwd{if} \CLParen\CFunName{\textcolor{red}{dispatch\_f}}\CLParen\CVar{f}, \CVar{acc}, \CVar{obsv}\CRParen\CRParen \CKwd{break};\\
			\qquad \CKwd{else} \CVar{acc} = \CFunName{\textcolor{red}{dispatch\_g}}\CLParen\CVar{g}, \CVar{acc}, \CVar{obsv}\CRParen;
		\;\;\\
		\quad \CKwd{return} \CVar{acc};\\\CRBrace
	\end{tabular}
	\caption{C implementation of the generic loop.}
	\label{fig:loopimpl}
\end{figure}
\subsection{Verifying C Parts of \bilby}
As mentioned, the previous verification of the functional correctness of key operations in \bilby~\cite{Amani_HCRCOBNLSTKMKH_16}
simply assumed the correctness of abstract functions, including an axiomatisation of array operations~\cite{Amani:phd}. 
We removed the axiomatisation for the five core array operations that we verified, 
and were able to show that the functional correctness proofs compose for the combined system. 
The array operations that remain unverified are functions like \texttt{create}, \texttt{set},
\texttt{copy}, \texttt{cmp}, that depend on platform-specific
functions such as \texttt{malloc}.
% and \texttt{modify}, which updates an element in
%an array by a function, but this can be implemented in terms of our verified
%\waget{} and \waput. % but is more efficient for arrays of non-primitive types --- for word arrays it is equivalent to \waput.
\begin{figure*}
	\begin{subfigure}[t]{0.475\textwidth}
		\small
		\begin{tabular}{l}
			\RepeatAtm{u} \CLParen\CVar{\mu_1}, (\CVar{n}, \CVar{f}, \CVar{g}, \CVar{acc}, \CVar{obs})\CRParen \CLParen\CVar{\mu_3}, \CVar{y}\CRParen =\\
				\quad \CLParen\CVar{n} > 0 $\longrightarrow$\\
					\qquad\quad$\exists$\CVar{b}.
						$\UpdSemAb{\xi_u}{[a \mapsto (\mathit{acc}, \mathit{obs})]}{f\ a\ }{\ \mu_1}{\text{\CVar{b}})\ }{\ \mu_1}$ $\land$\\
						\qquad\qquad\CLParen\CVar{b} $\longrightarrow$ \CVar{\mu_1} = \CVar{\mu_3} $\land$ \CVar{y} = \CVar{acc}\CRParen$\land$\\
						\qquad\qquad\CLParen$\lnot$\CVar{b} $\longrightarrow$ $\exists$\CVar{\mu_2} \CVar{acc'}.\\
							\qquad\qquad\qquad$\UpdSemAb{\xi_u}{[a \mapsto (\mathit{acc}, \mathit{obs})]}{g\ a\ }{\ \mu_1}{\text{\CVar{acc'}}\ }{\ \mu_2}$ $\land$\\
							\qquad\qquad\qquad\RepeatAtm{u} \CLParen\CVar{\mu_2}, (\CVar{n}-1, \CVar{f}, \CVar{g}, \CVar{acc'}, \CVar{obsv})\CRParen
								\CLParen\CVar{\mu_3}, \CVar{y}\CRParen\CRParen\CRParen$\land$\\
				\quad \CLParen\CVar{n} = 0 $\longrightarrow$
	\CVar{\mu_1} = \CVar{\mu_3} $\land$ \CVar{y} = \CVar{acc}\CRParen
		\end{tabular}
		\caption{C abstractions for the \textit{update} semantics.}
		\label{fig:uembedloop}
	\end{subfigure}\hfill
	\begin{subfigure}[t]{0.475\textwidth}
		\small
		\begin{tabular}{l}
			\RepeatAtm{v} \CLParen\CVar{n}, \CVar{f}, \CVar{g}, \CVar{acc}, \CVar{obs}\CRParen \CVar{y} =\\
				\quad\CLParen\CVar{n} > 0 $\longrightarrow$\\
					\qquad\quad$\exists$\CVar{b}.
						$\ValSemA{\xi_v}{[a \mapsto (\mathit{acc}, \mathit{obs})]}{f\ a\ }{\text{\CVar{b}})\ }$ $\land$\\
						\qquad\qquad\CLParen\CVar{b} $\longrightarrow$ \CVar{y} = \CVar{acc}\CRParen$\land$\\
						\qquad\qquad\CLParen$\lnot$\CVar{b} $\longrightarrow$ $\exists$\CVar{acc'}.\\
							\qquad\qquad\qquad$\ValSemA{\xi_v}{[a \mapsto (\mathit{acc}, \mathit{obs})]}{g\ a\ }{\text{\CVar{acc'}}\ }$ $\land$\\
							\qquad\qquad\qquad\RepeatAtm{v} \CLParen\CVar{n}-1, \CVar{f}, \CVar{g}, \CVar{acc'}, \CVar{obsv}\CRParen
								\CVar{y}\CRParen\CRParen$\land$\\
				\quad\CLParen\CVar{n} = 0 $\longrightarrow$ \CVar{y} = \CVar{acc}\CRParen
		\end{tabular}
		\caption{C abstractions for the \textit{value} semantics.}
		\label{fig:vembedloop}
	\end{subfigure}
	\caption{\Cogent-compatible abstraction for the generic loop.}
	\label{fig:embeddingloop}
\end{figure*}
\subsection{Binary Search}
\label{sec:binsearch}
\Cref{ex:binarySearch} gives the \Cogent code for binary search using our previously defined and verified 
array functions and $\mathsf{repeat}$. Note that while we specify the maximum number of iterations to $\mathsf{repeat}$ as the length of the array,
the algorithm is still $\mathcal{O}(\log n)$ as it will always exit early from the $\mathsf{stop}$ condition. 
This \emph{fuel} argument is an easy way to ensure that Isabelle is convinced that our functions terminate.

We firstly prove that the generated shallow embedding, which strongly resembles the \Cogent code, 
is correct:
\begin{figure}
	\small
	\begin{tabular}{l}
		\CKwd{type} \CTypeName{Range} = (\CPrimType{U32}, \CPrimType{U32}, \CPrimType{Bool})\\
		\CFunName{stop} : (\CTypeName{Range},
			((\CTypeName{Array} \CPrimType{U32})!, \CPrimType{U32})) \CFunctionArrow \CPrimType{Bool}\\
		\CFunName{stop} ((\CVar{l}, \CVar{r}, \CVar{b}), (\CVar{arr}, \CVar{v})) = $b \lor l \geq r$\\
		\CFunName{search} : (\CTypeName{Range},
			((\CTypeName{Array} \CPrimType{U32})!, \CPrimType{U32})) \CFunctionArrow
			\CTypeName{Range}\\
		\CFunName{search} ((\CVar{l}, \CVar{r}, \CVar{b}), (\CVar{arr}, \CVar{v})) =\\
		\quad \CKwd{let} \CVar{m} = \CVar{l} + (\CVar{r} - \CVar{l}) $\div$ 2 \CKwd{and}\\
		\qquad\ \CVar{x} = \CFunName{get} (\CVar{arr}, \CVar{m}, 0)\\
		\quad \CKwd{in}
		\CKwd{if}\; \!| \CVar{x} < \CVar{v} $\rightarrow$ (\CVar{m}+1, \CVar{r}, \CVar{b})\\
		\qquad\quad\, \!| \CVar{x} > \CVar{v} $\rightarrow$ (\CVar{l}, \CVar{m}, \CVar{b})\\
		\qquad\quad\, \!| \CKwd{else} $\rightarrow$ (\CVar{m}, \CVar{r}, \CKwd{True})\\

		\CFunName{binary\dash{}search} : ((\CTypeName{Array} \CPrimType{U32})!, \CPrimType{U32}) \CFunctionArrow \CPrimType{U32}\\
		\CFunName{binary\dash{}search} (\CVar{arr}, \CVar{v}) =\\
		\quad \CKwd{let} \CVar{len} = \CFunName{length} \CVar{arr} \CKwd{and}\\
		\qquad\ (\CVar{l}, \CVar{r}, \CVar{b}) = \CFunName{repeat} (\CVar{len},
			\CFunName{stop}, \CFunName{search}, (0, \CVar{len}, \CKwd{False}), (\CVar{arr}, \CVar{v}))\\
		\quad \CKwd{in} \CKwd{if}\;\! \CVar{b} \CKwd{then} \CVar{l} \CKwd{else} $\rightarrow$ \CVar{len}
	\end{tabular}
	\caption{The binary search algorithm in \Cogent.}
	\label{ex:binarySearch}
\end{figure}
\begin{theorem}[Correctness of Cogent binary search]\label{thm:binarySearchS}
	Let \CVar{i}= \CFunName{unat} \CLParen\CFunName{binary\dash{}search\textsubscript{s}} \CLParen\CVar{xs}, \CVar{v}\CRParen\CRParen in\\
	\begin{tabular}{l}
		\CFunName{sorted} \CVar{xs} $\land$ \CFunName{length}
		\CVar{xs} < $2^{32}$ $\longrightarrow$\\
		\quad \CLParen \CVar{i} < \CFunName{length} \CVar{xs} $\longrightarrow$ \CVar{xs} \CBang \CVar{i} = \CVar{v}\CRParen \\
		\quad $\land$ \CLParen$\lnot$ \CVar{i} < \CFunName{length} \CVar{xs} $\longrightarrow$ \CVar{v} $\notin$ \CFunName{set} \CVar{xs}\CRParen
	\end{tabular} 
	\end{theorem}
\noindent where \CFunName{length}, \CFunName{set}, and \CFunName{sorted} are Isabelle's list library functions that return the length of a list, turn a list to a set, and check whether a list is sorted, respectively. 
We denote search failure by returning an index that is out of bounds. 

From our overall refinement theorem (\Cref{thm:everything}), we can easily conclude that the C implementation is correct and remove any reference to Cogent:
\begin{corollary}[Correctness of the C binary search]\label{thm:binarySearchC}
	$ $\\
	Let \CVar{xs} be the list abstraction of the array \CVar{arr} for C heap $\sigma$,
	\CFunName{valid} \CLParen$\sigma$, \CVar{arr}\CRParen be the predicate that states that the array is valid,
	\ie the array's size (in bytes) is less than the size of memory ($2^{32}$ bytes) and is well-formed,
	and \CFunName{same} \CLParen$\sigma$, $\sigma'$, \CVar{arr}\CRParen be the predicate that states that an array is the same for the given C heaps:\\
	\begin{tabular}{l}
		\CFunName{sorted} \CVar{xs} $\land$
			\CFunName{valid} \CLParen$\sigma$, \CVar{arr}\CRParen $\land$\\
			$\CSem{(\text{\CVar{arr}}, \text{\CVar{v}})}{\sigma}{\text{\CFunName{binary\dash{}search\textsubscript{c}}}}{i}{\sigma'}$ $\longrightarrow$\\
		\quad \CFunName{valid} \CLParen$\sigma'$, \CVar{arr}\CRParen $\land$
			\CFunName{same} \CLParen$\sigma$, $\sigma'$, \CVar{arr}\CRParen $\land$\\
		\quad \CLParen\CFunName{unat} \CVar{i} < \CFunName{length} \CVar{xs} $\longrightarrow$ 
			\CVar{xs} \CBang \CLParen\CFunName{unat} \CVar{i}\CRParen = \CVar{v}\CRParen
		$\land$\\
		\quad \CLParen$\lnot$ \CFunName{unat} \CVar{i} < \CFunName{length} \CVar{xs} $\longrightarrow$ \CVar{v} $\notin$ \CFunName{set} \CVar{xs}\CRParen
	\end{tabular}
\end{corollary}
\noindent This theorem depends on no additional assumptions about any functions or any other part of the Cogent framework. With this, we have shown that the Cogent framework enables the verification of combined \Cogent-C systems.

\section{Related Work}
%\christine{Louis: write initial draft}
\paragraph*{Compiler Correctness}

Patterson and Ahmed~\cite{PattersonAhmed19} have defined a spectrum of compiler verification theorems focusing on compositional compiler correctness, extensible through linking. \cogent's certifying compiler does not neatly fall on this spectrum as the compiler itself does the linking of the \cogent-C system, and we are linking with manually verified C rather than other compiler outputs.
% Nevertheless, the FFI constraints we prove would remain the same if it were to link to other compiler outputs. 

%\christine{see PattersonAhmed19 paper's reference of Ceto being relevant to our FFI approach, titled: 
%Compiler Verification Meets Cross-Language Linking via Data
%Abstraction. (OOPSLA' 14).}
%\christine{this needs revision to double check what they've done, they also do something with ADTs}
Like \cogent, the Cito  language~\cite{WangCC14} allows cross-language linking and combined verification in a proof assistant. Unlike \cogent, Cito is a low-level C-like language without a sophisticated type system, so there are no FFI requirements to enforce static guarantees.
Pit-Claudel et al.~\cite{ijcar} use Cito as a target for verified compilation of relational queries, and support linking with foreign assembly code, but do not compose verification across languages for refinement from a common high-level spec, as we do.
%\christine{follow this up with something about Fiat if it is relevant (https://dspace.mit.edu/bitstream/handle/1721.1/107293/973557793-MIT.pdf?sequence=1)}
%say a compiler for a language designed for easing data structure verification.  
%We believe our approach relevant to compositional correctness and language interoperability and we
%Therefore, we expect the ideas used in our approach and in particular the way we set up our FFI, to be partially reusable in that context. 

%Although Cogent's compiler is a certifying compiler, and hence, is not a
%verified compiler, it still provides a mechanism to produce correct
%compilation of partial \cogent programs, which can be made into whole program
%compilation correctness by manually proving the functional correctness and
%the satisfaction of \cogent's FFI constraints for all defined abstract
%functions.
An important part of verified compilation is how the semantics
of linking is defined for the source language. The simplest way to
define linking on the source language is to require that 
a program may only be linked if it refines a program definable in the
source language. 
%This make linking trivial. 
SepCompCert~\cite{Kang_POPL_2016} and Pilsner~\cite{Neis_ICFP_2015} take
this approach. This approach is suitable for
languages which, unlike \Cogent, are expressive enough to encompass all programs, 
but we require C functions to refine more expressive specifications than what is definable natively in \Cogent. 
The \Cogent compiler instead generates shallow embeddings of \cogent programs in
Isabelle/HOL, that enable linking at the Isabelle specification level. Hence,
\Cogent supports source-independent linking insofar as we can embed all source languages 
inside Isabelle/HOL.
\paragraph*{Verification Approach}
%\begin{itemize}
%	\item hs-to-coq. Translate Haskell containers into Coq, and then verify
%		that these containers do in fact follow their spec. We are doing a
%		similar thing except we translate C to Isabelle/HOL through
%		intermediate abstractions. Our implementation is more low level
%		than Haskell containers and C has a very different style to
%		Isabelle as compared to Haskell and Coq.
%\end{itemize}
%Our verification approach is far from novel, and tried and tested approach of data refinement
%\cite{Breitner_SLRW_18}
%

Arrays are widely used data structures, particularly in low level systems programming. %\cite{}.
AutoCorres~\cite{Greenaway_AK_12,Greenaway_LAK_14} simplifies the verification of C code by abstracting it into a monadic embedding of C in Isabelle/HOL. The \cogent compiler further simplifies reasoning about the Cogent parts of a system by abstracting AutoCorres' monadic C into a purely functional Cogent embedding in Isabelle/HOL. For foreign functions, abstraction from AutoCorres' monadic C embedding to a purely functional  embedding is manually verified through intermediate Cogent embeddings.
AutoCorres  comes with a number of examples including binary search and QuickSort on 32-bit word arrays which use the \waget and \waput operations. Their proofs for these operations (roughly 250 lines) are about half the size of ours but ours are reusable for arrays with elements of any unboxed type.

If we were only concerned with proving functional correctness of an
array implementation, we could have taken a top down
approach, \ie generate an implementation from the specification, using
existing tools or frameworks~\cite{Lammich_ITP_2019, Lammich2019, Lammich_ITP_2010},  rather than prove that an
implementation refines its specification.
While a top down approach would likely lessen the
verification burden, our approach allows for more
control over the implementation and is suitable for
cases where an implementation already exists, as is the case for \bilby.

As mentioned in the introduction, \cogent purposely lacks recursion and iteration in order to ensure totality. There is an ongoing effort to add a limited form of recursion to \cogent~\cite{Murray:2019} while retaining totality through a termination checker~\cite{Murray:2019,Qiu:2020}. Extending the proof infrastructure of \cogent to support recursive types and to certify termination for functions that pass the termination checker is non-trivial and still remains open. 
Additionally, we are exploring adding arrays as built in types in \cogent and extending the compiler certificate to account for a number of array operations. Even then, we believe that the language interoperability approach presented in this paper remains valuable. It serves as a reference for how \cogent users and system developers can contribute their own additional data structures with their preferred implementation to \cogent rather than relying solely on \cogent developers to extend the language with each desired additional data structure and data structure operation. Moreover, operations that are directly implemented in C and called through the FFI may occasionally be implemented more efficiently than if they were directly implemented in \cogent. This is because users can escape the \cogent type system when implementing operations in C and for instance use internal aliasing within a function's implementation as long as they can prove that the overall C implementation respects the frame invariant. 

%It would be an interesting exercise, however, to investigate
%generating Cogent-C implementations using a top down
%approach.

\section{Conclusion}

We have demonstrated our cross-language approach to proving software correct.
Our systems mix \cogent, a safe functional language with a compiler that proves most of the required theorems automatically, and C, an unsafe imperative language with few guarantees. 
Specifically, we verified the array implementation and general loop iterators provided in \cogent's ADT library, which were used the implementation of real-world file systems, 
and we showed that they maintain the invariants required by \Cogent. This enabled us to eliminate some key assumptions in the pre-existing verification of the \bilby file system in Cogent.
These case-studies demonstrate that manual C verification can be straightforwardly composed with \cogent's  refinement chain, leading to a top-level shallow embedding that can be seamlessly connected with functional correctness specifications to ensure the correctness of an overall \cogent-C system.

\begin{acks}
This research is partially supported by an Australian Government Research Training Program (RTP) Scholarship.
\end{acks}

%%
%% Bibliography
%%

%% Please use bibtex, 

\bibliographystyle{plainnat}
\bibliography{references}

\end{document}